\documentclass[11pt]{article}

\usepackage[margin=1in]{geometry}
\usepackage{amsmath,amssymb,amsthm,mathtools}
\usepackage{enumitem}
\usepackage{booktabs}
\usepackage{float}
\usepackage{microtype}
\usepackage{xcolor, xspace}
\usepackage{hyperref}
\usepackage[nameinlink,capitalize,noabbrev]{cleveref}
\usepackage{thmtools}
\usepackage{thm-restate}
\usepackage{algorithm}

\newtheorem{theorem}{Theorem}[section]
\newtheorem{lemma}[theorem]{Lemma}
\newtheorem{proposition}[theorem]{Proposition}
\newtheorem{corollary}[theorem]{Corollary}

\theoremstyle{definition}

\newtheorem{remark}[theorem]{Remark}

\newcommand{\R}{\mathbb{R}}
\newcommand{\Z}{\mathbb{Z}}

\newcommand{\cut}{\delta}

\newcommand{\FAIL}{\operatorname{FAIL}}

\def\final{0}  
\def\iflong{\iffalse}
\ifnum\final=0  
\newcommand{\knote}[1]{{\color{blue}[{Karthik: \bf #1}]\marginpar{\color{blue}*}}}
\newcommand{\todo}[1]{{\color{red}[{\tiny TODO: \bf #1}]\marginpar{\color{red}*}}}
\newcommand{\znote}[1]{{\color{purple}[{Weihao: \bf #1}]\marginpar{\color{purple}*}}}
\newcommand{\cnote}[1]{{\color{magenta}[{Chandra: \bf #1}]\marginpar{\color{magenta}*}}}
\else 
\newcommand{\knote}[1]{}
\newcommand{\todo}[1]{}
\newcommand{\znote}[1]{}
\newcommand{\cnote}[1]{}
\fi  

\title{Multiobjective Hypergraph Min-Cut in Quasi-Polynomial Time\thanks{Grainger College of Engineering, University of Illinois, Urbana-Champaign, Email: {\tt\{karthe, chekuri, weihaoz3\}@illinois.edu}. Supported in part by NSF grant CCF-2402667. }}
\author{Karthekeyan Chandrasekaran \and Chandra Chekuri \and Weihao Zhu}
\date{}

\begin{document}
\maketitle

\begin{abstract}
We study the multiobjective hypergraph min-cut problem: Given a hypergraph \(H=(V,E)\) and $k$ cost functions $c_1, c_2, \ldots, c_k:E\to\Z_{\ge 0}$, the goal is to find a non-empty proper subset $U\subsetneq V$ of vertices with minimum $\max_{i\in [k]} c_i(\delta(U))$. 
When $k$ is part of input, the problem is NP-hard (even in graphs). We focus on fixed-constant $k$ setting (e.g., $k=1, 2, 3, \ldots$). Single-objective hypergraph min-cut as well as multiobjective graph min-cut for a constant number of objectives admit polynomial-time algorithms. In contrast to these special cases, the complexity of multiobjective hypergraph min-cut remains open even for $k=2$. Known techniques fail to extend due to structural differences between graphs and hypergraphs. For $k$-objective hypergraph min-cut when $k$ is a fixed constant, we design a randomized PTAS, and two different randomized quasi-polynomial time algorithms. As an application of our $2$-objective hypergraph min-cut results, we obtain a quasi-polynomial time approximation scheme (QPTAS) for hypergraph connectivity interdiction. AI tools were used to iterate and refine the algorithmic ideas underlying this work. 
\end{abstract}
\newpage
\setcounter{page}{1}
\section{Introduction}
We consider the multiobjective min-cut problem in hypergraphs. 
A hypergraph $H=(V, E)$ is specified by a vertex set $V$ and a hyperedge set $E$ where each hyperedge $e\in E$ is a subset of $V$. If all hyperedges have size at most $2$, then the hypergraph is simply a graph. 
The \emph{size} of a hypergraph is the sum of the sizes of its hyperedges and the number of vertices, denoted $p:=|V|+\sum_{e\in E}|e|$. 
For a non-empty proper subset $U\subsetneq V$, we denote $\delta_H(U):=\{e\in E: e\cap U\neq \emptyset \text{ and } e\setminus U\neq \emptyset\}$ to be the set of hyperedges crossing $U$. 
A subset $F\subseteq E$ of hyperedges is a \emph{cut} 
if $F=\delta_H(U)$ for some non-empty proper subset $U\subseteq V$. For a hyperedge-cost function $c: E\rightarrow \Z_{\ge 0}$ and a subset $F$ of hyperedges, we denote $c(F):=\sum_{e\in F}c(e)$. 
We define the multiobjective min-cut problem in hypergraphs below. 
\begin{center}
\fbox{%
\begin{minipage}{0.95\textwidth}
\textbf{$k$-Objective Hypergraph Min-Cut for Fixed Constant $k$.}

\textbf{Given:} Hypergraph \(H=(V,E)\) and costs \(c_1,c_2, \ldots, c_k:E\to\Z_{\ge0}\).

\textbf{Goal:} Find a cut $F$ with minimum $\max_{i\in [k]}c_i(F)$, i.e., 
\begin{equation}
        \min_{F\subseteq E:\ F \text{ is a cut}}
        \max_{i\in [k]} c_i(F).
        \tag{Min-Max-Cut}
        \label{eq:mm-cutset}
\end{equation}

\end{minipage}}
\end{center}

The problem can equivalently be phrased as a cut-shore formulation: find a non-empty proper subset $U\subsetneq V$ with minimum $\max_{i\in [k]}c_i(\delta(U))$, i.e., 
\begin{equation}
        \min_{\emptyset\ne U\subsetneq V}
        \max_{i\in [k]} c_i(\cut_H(U)).
        \tag{Min-Max-Cut-Shore}
        \label{eq:mm-shore}
\end{equation}
Several different cut-shores may have the same cut in a hypergraph. 
We work with the \eqref{eq:mm-cutset} formulation instead of the \eqref{eq:mm-shore} formulation, partly due to the status of the single-objective hypergraph min-cut problem. For the single-objective problem, it is well-known that the number of optimum solutions to \eqref{eq:mm-cutset} is at most $\binom{n}{2}$ (if the hypergraph is connected and all hyperedges have non-zero cost) and they can all be enumerated in polynomial time, but the number of optimum solutions to \eqref{eq:mm-shore} could be exponential in $n$, where $n$ is the number of vertices in the input hypergraph. 
The cut-shore formulation allows a different interpretation of the problem that is useful in applications. Suppose we have $k$ different hyperedge-weighted hypergraphs $H_1=(V,E_1,c_1),\ldots,H_k=(V,E_k,c_k)$ on the same vertex set; the hyperedge-sets are different and the cost $c_i : E_i \rightarrow \mathbb{R}$ is defined only on $E_i$ for each $i\in [k]$. We think of these hypergraphs as overlaid on each other. Then, the goal is to find  a non-empty proper subset $U \subsetneq V$ that minimizes $\max_{i\in [k]} c_i(\delta(U))$. It is not hard to see that this formulation is equivalent to the one we gave.

We choose to work with non-negative integral costs for ease of presentation in the introduction---our results hold for non-negative rational costs as well. 
When $k$ is part of the input, the problem is strongly NP-hard even in graphs \cite{AZ06}. We focus on the case of fixed constant $k$ throughout this work. 
We design a quasi-polynomial-time algorithm and a PTAS for $k$-objective hypergraph min-cut for every fixed constant $k$. 
We denote the special case of $k=1$ as Hypergraph Min-Cut and the case of $k=2$ as Bi-objective Hypergraph Min-Cut. We will mostly focus on Bi-objective Hypergraph Min-Cut throughout the introduction to illustrate the ideas, since that is the main bottleneck. 

\paragraph{Motivation.} In optimization, it is natural to have multiple different objectives, and for this reason, multiobjective combinatorial optimization arises in diverse settings, and has many applications. A classic example is the $s$-$t$ shortest path with hop and length constraints, and is motivated by applications in telecommunication networks; this corresponds to $k=2$ objectives. Some of the multiobjective problems are more naturally phrased as budgeted problems and we discuss this view subsequently. Length and hop-constrained network design is a very active area of recent research --- see \cite{hop_distance21,filtser22,hershkowitz2026planar,chekuri2026polylogarithmic} and references therein. When $k$ is larger, the $s$-$t$ shortest path with multiple objectives is also called the robust $s$-$t$ path problem --- see \cite{li2023polylogarithmic} for some recent work. 
Papadimitriou and Yannakakis connect multiobjective optimization to approximate Pareto optimal sets \cite{PY00} and also discuss several important examples and applications. Another motivation for multiple objectives comes from fairness considerations and there is a substantial recent literature in this direction with many different problems; we do not address them here. 
Armon and Zwick \cite{AZ06} were the first to consider multiobjective graph global min-cut. Two significant applications of the tractability of bi-objective min-cut in graphs are to: (i) graph connectivity interdiction problem--- Zenklusen obtained a PTAS for this problem by reducing it to bi-objective graph budgeted-cut \cite{Zen14} and (ii) flexible network design problem---verifying whether a solution is feasible to this problem reduces to bi-objective graph budgeted-cut \cite{AHM2022}. 
In this vein, multiobjective min-cut in hypergraphs is a fundamental problem for which tractability results will likely lead to interesting applications and fresh directions. We will derive some concrete applications of our results later.

\paragraph{Graphs.} 
Single-objective graph min-cut is a classical problem and is polynomial-time solvable via network flow techniques and also via submodular function minimization --- these techniques solve the min-$s$-$t$-cut problem to which one can reduce the (global) min-cut problem. 
However, bi-objective min-$s$-$t$-cut is NP-Hard \cite{PY00}, so these techniques do not apply for bi-objective graph min-cut. 
The graph cut function is symmetric submodular which leads to another approach to solve single-objective graph min-cut via a vertex-ordering technique \cite{NI92, Que98}, but this connection does not seem to be useful to solve bi-objective graph min-cut. 
Armon and Zwick \cite{AZ06} identified a scalarization approach to solve bi-objective graph min-cut in polynomial time. They showed that if $U\subseteq V$ is an optimum cut-shore, then $U$ is a $2$-approximate min-cut under the scalarized cost-function $c=c_1+c_2$, i.e., $c(\delta_H(U))\le 2\lambda_{H, c}$, where $\lambda_{H, c}:=\min\{c(\delta(A)): \emptyset\neq A\subsetneq V\}$. 
It is well-known that the number of $2$-approximate min-cuts in a connected graph is $O(n^4)$ and they can all be enumerated in polynomial time \cite{Kar93, NNI97, CQX20, BCW23}. Consequently, bi-objective graph min-cut can be solved in polynomial time. Moreover, the number of optimum solutions is $O(n^4)$ if the support graph of the scalarized cost function $c$, i.e., the graph $(V, \{e\in E: c(e)>0\})$ is connected. This approach extends naturally to $k$-objective graph min-cut to yield an $n^{O(k)}$-time algorithm. 

\paragraph{Hypergraphs.} 
The scalarization approach for multiobjective graph-min-cut extends to \emph{constant-rank}\footnote{The rank of a hypergraph is the size of the largest hyperedge in the hypergraph. Thus, graphs are rank-$2$ hypergraphs. A hypergraph is \emph{constant rank} if all hyperedges have size at most a constant that is independent of the number of vertices.} hypergraphs leading to a run-time that is exponential in the rank \cite{KK15, BCX23}. Our focus here is on arbitrary-rank hypergraphs. 
The status of bi-objective min-cut in arbitrary-rank hypergraphs has been open. 
The central reason for this is that the scalarization approach does not extend to Bi-objective Hypergraph Min-Cut:  If $F$ is an optimum cut, then it still follows that $F$ is a $2$-approximate min-cut under the scalarized cost-function $c=c_1+c_2$.  However, in contrast to graphs, the number of $2$-approximate min-cuts in a hypergraph can be exponential in the number of vertices/hyperedges (see \cite{CXY21} for an example). 

Bi-objective hypergraph min-cut differs substantially from single-objective hypergraph min-cut structurally too:  the number of optimum solutions to \eqref{eq:mm-cutset} could be exponential in the number of vertices. We describe a simple example to illustrate this: 
Consider the hypergraph on $n$ vertices obtained from the complete graph $K_n$ on $n$ vertices by adding a hyperedge $f:=V$ (i.e., it contains all vertices). Define costs as follows: $c_1(f)=\lfloor n^2/4\rfloor, c_2(f)=0$ and $c_1(e)=0, c_2(e)=1$ for every edge $e$ of size $2$. Then, the optimum value for \eqref{eq:mm-cutset} is $\lfloor n^2/4\rfloor$ and $\delta_H(U)$ for every non-empty proper subset $U$ of vertices is an optimum solution; thus, the number of optimum cuts is $2^{n-1}-1$. 
In contrast, it is well-known that the number of optimum solutions to single-objective hypergraph min-cut is at most $\binom{n}{2}$ (if the hypergraph is connected with all hyperedge-costs being non-zero) \cite{CX18, GKP17, CXY21, BCW23-enumhkcut-j}. 

The combinatorics of the Pareto-front of multiobjective min-cut in graphs and constant-rank hypergraphs has been extensively investigated in the last decade. 
\cite{AMMQ15} showed a strongly polynomial bound on the number of \emph{parametric min-cuts} for constant-rank hypergraphs---see the paper for the definition of parametric min-cuts. 
A recent series of works \cite{AMR17, Kar16, BCW23} designed variants of the uniform random contraction approach to solve and enumerate all Pareto-optimal cuts in graphs and constant-rank hypergraphs. This line of work naturally raises the question of whether multiobjective min-cut in arbitrary rank hypergraphs is tractable. Our work is a positive step towards answering this question. 

The existence of an instance with exponentially many optimum solutions is a key obstruction that illustrates the challenges in designing a sub-exponential time algorithm for bi-objective hypergraph min-cut. In fact, all known algorithms for bi-objective graph min-cut have also been powerful enough to solve the enumeration problem (i.e., enumerate all optimum solutions) along with the search problem. 
However, for bi-objective hypergraph min-cut, the existence of instances with exponentially many optimum solutions is a strong indication 
that algorithms for graphs cannot be adapted directly to solve the problem in hypergraphs. In particular, any sub-exponential-time algorithm for hypergraphs cannot afford to solve the enumeration problem. 
The budgeted perspective of bi-objective hypergraph min-cut that we define next helps focus on the search problem by side-stepping the enumeration problem. 

\paragraph{$k$-Objective Budgeted-Cut.} Let $H=(V, E)$ be a hypergraph with $k$ non-negative hyperedge-cost functions $\bar{c}=(c_1\in \Z^E_{\ge 0}, c_2\in \Z^E_{\ge 0}, \ldots, c_k \in \Z^E_{\ge 0})$ and budgets $\bar{b}=(b_1, b_2, \ldots, b_k)\in \Z^k_{\ge 0}$ that we collectively denote as $(H, \bar{c}, \bar{b})$. 
We recall that a cut is a subset $F\subseteq E$ of hyperedges where $F=\delta_H(U)$ for some non-empty proper subset $U$ of vertices. 
A cut $F\subseteq E$ is a \emph{budget-feasible cut} for $(H, \bar{c}, \bar{b})$ if $c_i(F)\le b_i$ for every $i\in [k]$ and is an \emph{$\alpha$-approximate budget-feasible cut} for $(H, \bar{c}, \bar{b})$ for some positive real value $\alpha$ if $c_i(F)\le \alpha b_i$ for every $i\in [k]$. 
We define the $k$-Objective Hypergraph Budgeted-Cut problem below. 
\begin{center}
\fbox{%
\begin{minipage}{0.95\textwidth}
\textbf{$k$-Objective Hypergraph Budgeted-Cut for Fixed Constant $k$.}

\textbf{Given:} Hypergraph \(H=(V,E)\), 

\quad \quad \quad \quad hyperedge-cost functions $\bar{c}=(c_1\in \Z^E_{\ge 0},c_2\in \Z^E_{\ge 0}, \ldots, c_k\in\Z^E_{\ge0})$, and 

\quad \quad \quad \quad budget \(\bar{b}=(b_1,b_2, \ldots, b_k)\in \Z^k_{\ge 0}\).

\textbf{Goal:} Return a budget-feasible cut for $(H, \bar{c},\bar{b})$ if one exists. 
\end{minipage}}
\end{center}
We denote the input instance of $k$-Objective Hypergraph Budgeted-Cut as $(H, \bar{c}, \bar{b})$. 
$k$-Objective Hypergraph Min-Cut reduces to $k$-Objective Hypergraph Budgeted-Cut by binary searching for the smallest $\lambda$ for which the diagonal budget vector $\bar{b}=(b_1=\lambda, b_2=\lambda, \ldots, b_k=\lambda)$ admits a budget-feasible cut. It is easy to construct graph instances where the number of single-objective budget-feasible cuts is exponential in the number of vertices (e.g., complete graph with unit cost function and budget being $\binom{n}{2}$). Thus, polynomial-time algorithms to solve $k$-Objective Hypergraph/Graph Budgeted-Cut (for fixed $k$) cannot afford to solve the enumeration problem.

\subsection{Results}
Our main algorithmic results are a randomized quasi-polynomial time algorithm and a randomized PTAS for $k$-Objective hypergraph budgeted-cut. 

\begin{restatable}[Quasi-polynomial Time Algorithm]{theorem}{MultiQPTheorem}\label{thm:k-objective-qp}
There is a randomized algorithm that takes as input an instance $(H, \bar{c}, \bar{b})$ of $k$-Objective Hypergraph Budgeted-Cut and has the following properties: (i) if the instance has a budget-feasible cut, then the algorithm returns a budget-feasible cut with probability at least $1/n^{2k}$, and otherwise, returns FAIL and 
(ii) the expected run-time of the algorithm is $(m+1)^{5k\log_2{n}}$, where $m$ and $n$ are the number of hyperedges and the number of vertices of the input hypergraph. 
\end{restatable}

Repeating the algorithm from \Cref{thm:k-objective-qp} $O(n^{2k}\log(1/\eta))$ times finds a budget-feasible cut with probability at least $1-\eta$, whenever one exists. Consequently, we obtain an algorithm for $k$-objective Hypergraph Budgeted-Cut that runs in expected time
\begin{align}
        &m^{O(k\log n)}\log{\frac{1}{\eta}}. \label{qp-time}
\end{align}
Thus, we obtain a quasi-polynomial time algorithm for $k$-objective hypergraph budgeted-cut for every fixed constant $k$. 

\begin{restatable}[PTAS]{theorem}{multiptas}\label{thm:k-ptas}
For every $\varepsilon\in (0,1)$, there is a randomized algorithm 
that takes as input an instance $(H, \bar{c}, \bar{b})$ of $k$-Objective Hypergraph Budgeted-Cut and returns either FAIL or a $(1+\varepsilon)$-budget-feasible cut of $(H, \bar{c}, \bar{b})$ and has the following properties  (i) if the instance has a budget-feasible cut, then the algorithm returns a $(1+\varepsilon)$-budget-feasible cut with probability $1/n^{O(k\cdot \log{(k/\varepsilon)})}$ and if the instance has no $(1+\varepsilon)$-budget-feasible cut, then the algorithm returns FAIL and (ii) the run-time of the algorithm is $2^{O(k)}\cdot p^{O(1)}$, where $p$ is the size of the input hypergraph. 
\end{restatable}

Repeating the algorithm from \Cref{thm:k-ptas} $n^{O(k\log{(k/\varepsilon)})}\log{1/\eta}$ times will find a $(1+\varepsilon)$-approximate budget-feasible cut with probability at least $1-\eta$, whenever the input has a budget-feasible cut, in time
\[
       n^{O\left(k\log{\frac{k}{\varepsilon}}\right)}p^{O(1)}\log{\frac{1}{\eta}}.
\]
Thus, we obtain a PTAS for $k$-objective hypergraph budgeted-cut for every fixed constant $k$. 
We note that Theorem \ref{thm:k-ptas} also implies a quasi-polynomial time algorithm for $k$-objective hypergraph budgeted-cut for every fixed constant $k$: for $\varepsilon=\min\{1/2,\min_{i\in [k]:b_i>0}\{1/(2b_i)\}\}$, every $(1+\varepsilon)$-approximate budget-feasible cut is indeed a budget-feasible cut, since all costs are non-negative integers. 
Thus, we obtain an algorithm for $k$-objective Hypergraph Budgeted-Cut that runs in time 
\begin{align}
n^{O\left(k\max_{i\in [k]}\log{(kb_i)}\right)}p^{O(1)}\log{\frac{1}{\eta}}. \label{weak-qp-time}
\end{align}

It is helpful to keep in mind that $m\le p\le mn$ while comparing the run-times \eqref{weak-qp-time} and \eqref{qp-time}. 
We emphasize an advantage of the algorithm from Theorem \ref{thm:k-objective-qp} over the algorithm from Theorem \ref{thm:k-ptas}: the runtime of the algorithm from Theorem \ref{thm:k-ptas} stated in (\ref{weak-qp-time}) is only weakly quasi-polynomial, i.e., it depends on the magnitude of the budgets, while the runtime of the algorithm from Theorem \ref{thm:k-objective-qp} stated in (\ref{qp-time}) is strongly quasi-polynomial. 
The algorithms underlying Theorems \ref{thm:k-objective-qp} and \ref{thm:k-ptas} are different, although both rely on a combination of contraction and deletion ideas. Our algorithms borrow some ideas from previous random contraction algorithms and combine them in subtle ways. 

\paragraph{Hedgegraphs.} 
One of the side results from our approach is that it gives a new algorithm for the hedgegraph min-cut problem. A \emph{hedgegraph} is specified by a vertex set and a collection of hedges, where each hedge is a vertex-disjoint collection of hyperedges. In the hedgegraph min-cut problem, we are given a hedgegraph with a non-negative cost function $c$ on the hedges, and the goal is to find a minimum-cost subset of hedges whose removal disconnects the hedgegraph. 
Our approach, tailored for the single-objective problem, extends to give a new algorithm for hedgegraph min-cut that is (1) strongly quasi-polynomial time, thus complementing the known weakly quasi-polynomial time algorithm of Ghaffari, Karger, and Panigrahi \cite{GKP17} and  
(2) fixed-parameter tractable (FPT) parameterized by the solution size, thus complementing the known FPT \cite{FGKLS25} via an arguably simpler and more natural algorithm. See Remark \ref{rem:hedgegraph-results}. Another advantage of this algorithm is the unified perspective---the same algorithm achieves both quasi-polynomial time and FPT parameterized by solution size. 
Our quasi-polynomial time and PTAS results also extend to multiobjective hedgegraph budgeted-cut. 
We are omitting the proofs of these additional results in this paper to keep the focus on multiobjective hypergraph min-cut. There is still hope for a polynomial-time algorithm for multiobjective hypergraph min-cut, while a polynomial-time algorithm is unlikely even for single-objective hedgegraph min-cut \cite{JLMPTS26}.

\subsection{Applications} 
We discuss a few applications of our results for multiobjective hypergraph budgeted-cut. 
All of these applications were already known for graphs via reductions to multiobjective graph budgeted-cut. We note that these reductions hold verbatim for hypergraphs as well. Thus, our results, in conjunction with the reduction, lead to the discussed applications in a straightforward manner. We skip proofs of these applications to avoid repeating the literature. 

\paragraph{Min-constrained-cut in Hypergraphs.} 
Given a hypergraph $H=(V, E)$ with non-negative hyperedge costs $c:E\rightarrow \Z_{\ge 0}$ and an integer $r$, there exists a quasi-polynomial time algorithm to (1) find a minimum cut with at most $r$ hyperedges, i.e., $\min\{c(\delta(U)): \emptyset \neq U\subsetneq V, |\delta(U)|\le r\}$ and (2) find a minimum cut with at most $r$ vertices on the the smaller side of the cut, i.e., $\min\{c(\delta(U)): \emptyset\neq U\subsetneq V, \min\{|U|, |V\setminus U|\}\le r\}$. (1) easily reduces to the budgeted $2$-objective min-cut. (2) is more interesting, and less obvious, and the reduction for this (and also (1)) are given by Armon and Zwick \cite{AZ06} in the setting of graphs. The idea is to overlay a second complete graph on the original graph (with the same vertex set) and use separate budgets for these two graphs; each graph corresponds to a different cost. A budget of $r(n-r)$ on the edges of the complete graph ensures that the cut has at most $r$ vertices on the smaller side. The same reduction holds verbatim for hypergraphs; thus, our quasi-polynomial time algorithm for bi-objective hypergraph min-cut gives the above-mentioned results. 

\paragraph{Approximate Pareto Front.} 
A natural way to deal with multiobjective combinatorial optimization problems is to enumerate the 
Pareto-front. Let $H=(V, E)$ be a hypergraph with hyperedge-cost functions $\bar{c}=(c_1\in \Z^E_{\ge 0}, c_2\in \Z^E_{\ge 0}, \ldots, c_k\in \Z^E_{\ge 0})$. 
A set $P\subseteq \R^k$ is 
an \emph{$\epsilon$-approximate Pareto-front} if (1) for every point $(b_1, b_2, \ldots, b_k)\in P$, there exists a cut $F$ such that $c_i(F)=b_i$ for all $i\in [k]$ and $(2)$ for every cut $F$, there exists a point $(b_1, b_2, \ldots, b_k)\in P$ such that $b_i\le (1+\epsilon)c_i(F)$ for all $i\in [k]$. A $0$-approximate Pareto-front is known as the \emph{Pareto-front}. 
Theorem \ref{thm:k-ptas} directly implies a PTAS to enumerate the approximate Pareto-front --- run the PTAS for multiobjective hypergraph budgeted-cut for all budgets that are powers of $(1+\varepsilon)$ that are at most $c_i(E)$ for all $i\in [k]$. 
Papadimitriou and Yannakakis \cite{PY00} investigated the computational complexity of several multiobjective problems from the perspective of enumerating the Pareto-front. They showed a general result that if the \emph{exact} version of the single-objective problem is solvable in pseudo-polynomial-time, then there exists an FPTAS to enumerate an approximate Pareto-front for the multiobjective variant (which also implies an FPTAS for the multiobjective problem). Note that an FPTAS is much stronger than a PTAS. The relevant exact version of the single-objective graph min-cut problem is not solvable in pseudo-polynomial time since max-cut is strongly NP-hard, and consequently, an FPTAS for multiobjective graph/hypergraph min-cut does not follow directly from the results of \cite{PY00}.

\paragraph{Hypergraph Connectivity Interdiction.} 
The connectivity of a hypergraph $H=(V, E)$ with hyperedge capacities $c: E\rightarrow \Z_{\ge 0}$ is $\lambda_{H, c}:=\min\{c(\delta_H(U)): \emptyset\neq U\subsetneq V\}$. 
In hypergraph connectivity interdiction, we are also given costs $w: E\rightarrow \Z_{\ge 0}$ and a budget $B\in \Z_{\ge 0}$ and the goal is to find a subset $F$ of hyperedges of cost at most $B$ such that the deletion of $F$ minimizes connectivity. 
We formally define the hypergraph connectivity interdiction problem below.
\begin{center}
\fbox{%
\begin{minipage}{0.95\textwidth}
\textbf{Hypergraph Connectivity Interdiction.}

\textbf{Given:} Hypergraph \(H=(V,E)\), capacities $c: E\rightarrow \Z_{\ge 0}$, costs $w: E\rightarrow \Z_{\ge 0}$, and $B\in \Z_{\ge 0}$.

\textbf{Goal:} $\min\{\lambda_{H-F,c}: F\subseteq E, w(F)\le B\}$.
\end{minipage}}
\end{center}
It is easy to see that connectivity interdiction is NP-hard even on two vertex graphs by reduction from the knapsack problem. We will denote the special case of the problem with unit removal costs (i.e., $w(e)=1$ for all hyperedges $e$) as hypergraph cardinality connectivity interdiction. We say that an algorithm returns an \emph{$(\alpha, \beta)$-approximation} if it returns a subset $F$ of hyperedges such that $\lambda_{H-F, c}\le \alpha OPT$ and $w(F)\le \beta B$, where OPT is the optimum value of the instance. 

Zenklusen \cite{Zen14} investigated the approximability of graph connectivity interdiction and showed two reductions: 
(1) graph cardinality connectivity interdiction reduces to bi-objective graph budgeted-cut and 
(2) graph connectivity interdiction admits a randomized PTAS if bi-objective graph budgeted-cut is solvable in polynomial time. Since bi-objective graph budgeted-cut is indeed solvable in polynomial time, his reductions imply a polynomial time algorithm for cardinality connectivity interdiction and a PTAS for connectivity interdiction in graphs. We note that his reductions extend verbatim to hypergraph connectivity interdiction as well. 
His reductions to bi-objective hypergraph budgeted-cut in conjunction with our Theorem \ref{thm:k-objective-qp} lead to the first two corollaries below, and in conjunction with our Theorem \ref{thm:k-ptas} lead to the third corollary below. 

\begin{corollary}
There exists an algorithm for (where $p$, $m$, and $n$ denote the size, number of hyperedges, and number of vertices in the input hypergraph): 
\begin{enumerate}
\item 
hypergraph cardinality connectivity interdiction that returns a $(1, 1)$-approximation and runs in expected time $m^{O(\log{n})}p^{O(1)}$ (i.e., hypergraph cardinality connectivity interdiction admits a quasi-polynomial-time algorithm). 
\item 
hypergraph connectivity interdiction that returns a $(1+\varepsilon, 1)$-approximation and runs in expected time $m^{O(\frac{1}{\varepsilon}+\log{n})}p^{O(1)}$ (i.e., hypergraph connectivity interdiction admits a QPTAS).
\item hypergraph connectivity interdiction that returns a $(1+\varepsilon, 1+\varepsilon)$-approximation and runs in time $m^{O(\frac{1}{\varepsilon})}n^{O(\log{\frac{1}{\varepsilon}})}p^{O(1)}$ (i.e., hypergraph connectivity interdiction admits a bicriteria-PTAS). 
\end{enumerate}
\end{corollary}
A recent work has designed an FPTAS for graph connectivity interdiction by exploiting the polynomial-time enumerability of constant-approximate min-cuts in graphs \cite{HOY26}. We leave it as an open question to design a PTAS/FPTAS for hypergraph connectivity interdiction. 
\subsection{Techniques}\label{sec:techniques}
We highlight the ideas underlying the quasi-polynomial time algorithm stated in \Cref{thm:k-objective-qp} for the special case of $k=2$, i.e., for Bi-objective Hypergraph Budgeted-Cut since it carries all the necessary ideas needed to generalize to all $k\ge 1$. The algorithm is based on random contraction and deletion. 
We may assume that the input instance has no loops (since loops are not part of any cut) and no spanning hyperedges (otherwise, delete all spanning hyperedges from the instance and reduce the budget by the cost of the spanning hyperedges). We discuss the algorithm for the case of single-objective budgeted-cut problem first since it will illustrate two of the three ingredients needed for the Bi-objective Budgeted-Cut problem.

\paragraph{Single-objective Budgeted-Cut.}
Let $(H=(V, E), c\in \R^E_{\ge 0}, b\in \R_{\ge 0})$ be the input instance. 
The algorithm begins by verifying if there exists a singleton cut $\delta_H(\{v\})$ that is budget-feasible and if so, then returns it. Otherwise, every singleton cut has a cost greater than the budget $b$. 
The algorithm then samples a hyperedge $e$ with probability proportional to $c(e)$. It first recurses on the contracted instance $H/e$ with the same costs on surviving hyperedges and with the same budget $b$, termed the \emph{contraction call}. If the contraction call succeeds with a budget-feasible cut $F_1$, then the algorithm returns $F_1$. If the contraction call fails, then with probability $|e|/|V|$, the algorithm recurses on the deletion instance $H\setminus e$ with the same costs on the surviving hyperedges and with residual budget $b-c(e)$, termed as a \emph{deletion call}. If the deletion call succeeds with a budget-feasible cut $F_2$, then the algorithm returns a budget-feasible cut contained within $F_2\cup \{e\}$. Else, the algorithm returns FAIL. See Algorithm 3 in Appendix \ref{sec:uni} for a pseudocode.

The above algorithm closely resembles the branching-contractions algorithm for hypergraph min-cut due to Fox, Panigrahi, and Zhang \cite{FPZ23}---the main difference\footnote{Their algorithm is stated for Hypergraph Min-Cut but it can be viewed through the lens of Hypergraph Budgeted-Cut in a natural fashion.} is that when the contraction call fails, our algorithm recurses on the deletion instance $(H\setminus e, c, b-c(e))$ while their algorithm recurses on the same instance $(H, c, b)$, both with the same probability $|e|/|V|$. Clearly, our algorithm has fewer nodes in the recursion tree than their algorithm but their algorithm has better success probability than our algorithm. They showed that their algorithm has $O(n^2)$ expected nodes in the recursion tree and succeeds with probability $\Omega(\frac{1}{\log{n}})$. We are able to show that our algorithm has $O(n^2)$ expected nodes in the recursion tree and succeeds with probability $\Omega(\frac{1}{n^2})$. Although our success probability is worse, the advantage is that the algorithm generalizes to multiple objectives as we will discuss shortly. 

The bound on the recursion tree size is obvious since our recursion tree cannot have more nodes than that of \cite{FPZ23}---see Appendix \ref{sec:uni} for an analysis. We outline the success probability analysis to highlight its key ingredient. Let the instance be budget-feasible with $F\subseteq E$ being a budget-feasible cut. We need to show that the recursive instance preserves budget-feasibility with large probability (which will allow us to use induction since the number of vertices+hyperedges decreases in recursive calls). 
Note that the only interesting case is when all singleton cuts are budget-infeasible, i.e., $c(\delta(v))>b\ge c(F)$ for all $v\in V$. Adding $c(\delta(v))>c(F)$ across all $v\in V$ and via simple counting manipulations, we obtain the following inequality (see Lemma \ref{lem:uni-singleton-infeasibility-inequality}): 
\[
\sum_{e\in E\setminus F}c(e)|e|>\sum_{e\in F}c(e)|V\setminus e|. 
\]
We observe that this inequality can be rephrased as 
\[
\sum_{e\in E\setminus F}p_e\left(\frac{|e|}{|V|}\right)> \sum_{e\in F}p_e\left(1-\frac{|e|}{|V|}\right),
\]
where $p_e:=c(e)/c(E)$ is the probability of sampling hyperedge $e$. Note that the RHS is the probability of picking a hyperedge $e\in F$ and \emph{not} entering the deletion call while the LHS is \emph{smaller than} the probability of picking a hyperedge $e\in E\setminus F$ (and entering the contraction call). 
From the perspective of success probability, the RHS is the probability of the bad event while the LHS is \emph{smaller than} the probability of a good event. Thus, the probability of the bad event is smaller than that of a good event. This is the key observation that suggests that a charging analysis should work to show that the success probability is large. The rest is algebra. See Lemma \ref{lem:uni-success} for technical details. As mentioned earlier, our approach for the single-objective problem leads to two additional results on hedgegraph min-cut via the same algorithm: namely, a strongly quasi-polynomial-time algorithm improving on the weakly quasi-polynomial-time algorithm \cite{GKP17} and a fixed-parameter algorithm parameterized by the solution size that is simpler than known \cite{FGKLS25}---see Remark \ref{rem:hedgegraph-results}. 

\paragraph{Bi-objective Budgeted-Cut.} 
Let $(H=(V, E), \bar{c}=(c_1\in \R^E_{\ge 0}, c_2\in \R^E_{\ge 0}), \bar{b}=(b_1, b_2)\in \R^2_{\ge 0})$ be the input instance. Recall that we may assume that the input instance has no loops and no spanning hyperedges. The algorithm again begins by verifying if there exists a singleton cut $\delta_H(\{v\})$ that is budget-feasible on both objectives, and if so,  returns it and terminates. Otherwise, for every vertex $v\in V$, there exists $i\in [2]$ such that $c_i(\cut_H(\{v\}))>b_i$. For each $i\in [2]$, define the singleton-infeasible set
\[
        I_i:=\{v\in V:c_i(\cut_H(\{v\}))>b_i\}.
\]
Then $I_1\cup I_2=V$, so one of the two singleton-infeasible sets has size at least $|V|/2$.  The algorithm chooses an index $j\in[2]$ with larger $|I_j|$, in particular with 
\[
        |I_j|\ge \frac{|V|}{2},
\]
and samples a hyperedge with probability proportional to its $c_j$-cost.
Next, the algorithm recurses on the contracted instance \(H/e\) with the same costs on surviving hyperedges and with the same budget $\bar{b}$, termed a \emph{contraction call}.  If the contraction call succeeds with a budget-feasible cut $F_1$, then the algorithm returns $F_1$.  If the contraction call fails, then with probability 
\[
        \frac{|e\cap I_j|}{|I_j|},
\] 
the algorithm recurses on the deletion instance \(H\setminus e\) with the same costs on the surviving hyperedges and with residual budgets $(b_1-c_1(e), b_2-c_2(e))$, termed a \emph{deletion call}. If the deletion call succeeds with a budget-feasible cut $F_2$, then the algorithm returns a budget-feasible cut contained within $F_2\cup \{e\}$. Else, the algorithm returns FAIL. 

The main difference from the single-objective budgeted-cut algorithm is the sampling distribution and the probability of executing the deletion call. The sampling distribution is uniform with respect to the objective $j\in [2]$ with the larger singleton-infeasible set and the probability of executing the deletion call depends on the overlap of the sampled hyperedge with the larger singleton-infeasible set. The idea of sampling based on the size of the larger singleton-infeasible set has appeared in the context of multiobjective min-cuts for graphs and \emph{constant-rank} hypergraphs before---see \cite{Kar16, AMR17, BCX23}. However, the algorithm in those works is based on uniform\footnote{By uniform distribution, we mean that the sampling distribution is independent of the size of the hyperedges and depend only on the cost of the hyperedges. } random pure-contraction\footnote{By pure-contraction, we mean that the algorithm does not explore the deletion branch and recurses only on the contracted instance.} without exploring the deletion branch. Since we are dealing with arbitrary-rank hypergraphs, uniform random pure-contraction algorithms do not seem to work and this is the reason for exploring the deletion branch with non-zero probability even in the single-objective case. Note that the single-objective algorithm described earlier explores the deletion branch with probability $|e|/|V|$ whereas the bi-objective algorithm described here explores the deletion branch with a different probability choice that depends on the overlap of the sampled hyperedge with the larger singleton-infeasible set as well as the size of the larger singleton infeasible set, namely $\frac{|e\cap I_j|}{|I_j|}$. This probability choice to explore the deletion branch in the multiobjective setting had not appeared in previous works and it turns out to be crucial in the success probability analysis. 

We show that the algorithm succeeds with probability $\Omega(n^{-4})$. 
The success probability analysis closely follows that of the single-objective algorithm. Let $F$ be a budget-feasible cut. Again, the only interesting case is when all singleton cuts are budget-infeasible, i.e., for each $v\in V$, there exist $i\in [2]$ such that $c_i(\delta(v))>b_i\ge c_i(F)$. Thus, $I_1\cup I_2=V$. Let $j\in [2]$ such that $|I_j|\ge |V|/2$. Adding $c_j(\delta(v))>c_j(F)$ for all $v\in I_j$ and via simple counting manipulations, we obtain the following inequality (see \Cref{coro:singleton-infeasibility-ineq})
\[
\sum_{e\in E\setminus F}c_j(e)|I_j\cap e| >\sum_{e\in F}c_j(e)|I_j\setminus e|. 
\]
We observe that this inequality can again be rephrased as 
\[
\sum_{e\in E\setminus F}p_e\left(\frac{|e|}{|V|}\right)> \sum_{e\in F}p_e\left(1-\frac{|e\cap I_j|}{|I_j|}\right),
\]
where $p_e:=c_j(e)/c_j(E)$ is the probability of sampling hyperedge $e$. Again, from the perspective of success probability, the RHS is the probability of the bad event while the LHS is smaller than the probability of a good event. Thus, the probability of the bad event is at most the probability of a good event, which suggests a charging argument again. 
The other crucial ingredient is the choice of $j$, which ensures that $|I_j|\ge |V|/2$. See Lemma \ref{lem:multi-success} for details. 

We show that the expected size of the recursion tree is $m^{O(\log{n})}$, i.e., quasi-polynomial, by relying on the choice of $j\in [k]$, which ensures that $|I_j|\ge |V|/2$---see Lemma \ref{lem:multi-runtime} for details. The quasi-polynomial bound here is in contrast to that of the $O(n^2)$ bound for the single-objective algorithm. We believe that the quasi-polynomial bound is not just an artifact of the analysis---in particular, the bi-objective algorithm after sampling a hyperedge could potentially create the deletion branch with probability $1$ (since $e\cap I_j$ could be equal to $I_j$) whereas the single-objective algorithm creates the deletion branch with probability $|e|/|V|$ which is strictly less than $1$ (since the input has no spanning hyperedges); this subtlety is the reason for the weaker quasi-polynomial bound on the recursion tree size and we believe that this bound is tight for this algorithm. 

\paragraph{PTAS.}
Ghaffari, Karger, and Panigrahi \cite{GKP17} designed a PTAS for single-objective hypergraph min-cut. Their algorithm is based on random contraction and deletion ideas. 
If the input hypergraph has no large hyperedges, then it performs uniform random contraction. Otherwise, it branches with equal probability into one of two possibilities: (i) delete all large hyperedges by paying for them and (ii) sample a uniform random large hyperedge and recurse on the contracted instance. 
The PTAS result stated in Theorem \ref{thm:k-ptas} is obtained by adapting their algorithm to the budgeted-case, then incorporating the above-mentioned large singleton-infeasible set sampling approach if the hypergraph has no large hyperedges, and branching into one of three possibilities otherwise: (i) delete all large hyperedges by paying for them in the budget, (ii) sample a uniform random large hyperedge $e\in L$ with probability proportional to $c_1(e)$ and recurse on the contracted instance and (iii) sample a uniform random large hyperedge $e\in L$ with probability proportional to $c_2(e)$ and recurse on the contracted instance. See Section \ref{sec:multi-ptas} for details.

\subsection{AI Disclosure}\label{sec:ai-disc}
We used AI (ChatGPT 5.5 Pro and Extended Pro) to rule out, analyze, and refine various approaches over multiple iterations. AI was helpful for a quick analysis of variants of random contraction/deletion based algorithms. The quasi polynomial-time algorithm and analysis presented in Section \ref{sec:multi-quasi} are adaptations of the ones given by the AI after multiple rounds of ruling out polynomial-time algorithmic approaches and our suggestions to refine the approaches towards a quasi-polynomial-time algorithm. We subsequently designed the PTAS in Section \ref{sec:multi-ptas} and analyzed it independently. We used AI to typeset various technical details and rewrote them for clarity and readability. Authors assume full responsibility for all content.

\subsection{Preliminaries}
Let $H=(V, E)$ be a hypergraph. 
A loop is a hyperedge of size at most one. 
A spanning hyperedge is a hyperedge that contains the entire vertex set.  
Let $e\in E$. 
The contracted hypergraph, denoted $H/e$, is the hypergraph obtained from $H$ by identifying all vertices in \(e\) into a single vertex and deleting the hyperedge $e$. 
The deletion hypergraph, denoted $H\setminus e$, is the hypergraph obtained from $H$ by deleting the hyperedge $e$. 
For a subset $L\subseteq E$, we denote the hypergraph obtained by deleting $L$ as $H-L$. 
Let $\bar{c}=(c_1\in \R^E_{\ge 0}, c_2\in \R^E_{\ge 0}, \ldots, c_k\in \R^E_{\ge 0})$ be hyperedge-cost functions and $\bar{b}=(b_1, b_2, \ldots, b_k)\in \R^k_{\ge 0}$ be the budget vector. For a subset $A\subseteq E$, we denote $\bar{b}-\bar{c}(A):=(b_1-c_1(A), b_2-c_2(A), \ldots, b_k-c_k(A))$. For a subset $A\subseteq E$, the cost function $c_i$ restricted to the hyperedges in $E-A$ is denoted as $c_i|_{E-A}$ for every $i\in [k]$ and define $\bar{c}|_{E-A}:=(c_1|_{E-A}, c_2|_{E-A}, \ldots, c_k|_{E-A})$. For notational convenience, we write $E-e$ to denote $E-\{e\}$.
In our algorithms, we follow the convention that 
$E$ is an indexed multiset of hyperedges. Under contraction, each surviving
hyperedge identity is mapped to its image in the contracted vertex set; parallel images
are retained as distinct identities, or equivalently may be merged only after summing all
cost coordinates and recording the corresponding original identities. 
Recursive outputs are interpreted as sets of original hyperedge identities. 
Throughout, $\log{}$ is to the base $2$. 
We will use the following proposition.

\begin{proposition}\label{prop:bi-numeric}
For all integers $n\ge3$, $2\le a\le n-1$, and $t\ge 1$, we have the following
\begin{align}
        \left(\frac{n}{n-a+1}\right)^t-1
        &\ge
        \frac{ta}{2n}. \label{prop:bi-numeric-ineq}
\end{align}
\end{proposition}

\begin{proof}

Let $x=(a-1)/n$.  Then, $0<x<1$ and
\[
        \frac{n}{n-a+1}=\frac{1}{1-x}.
\]
Using $(1-x)^{-t}\ge 1+tx$, we get
\[
        \left(\frac{n}{n-a+1}\right)^t-1
        \ge
        \frac{t(a-1)}{n}
        \ge
        \frac{ta}{2n},
\]
where the last inequality uses $a\ge2$.
\end{proof}

\section{Singleton Infeasible Sets}
In this section, we focus on budget-infeasible singleton cuts and prove a useful inequality.  This is one of the key ingredients in the correctness analysis of the algorithms for both \Cref{thm:k-objective-qp,thm:k-ptas}. The following lemma is the budgeted variant of the statement that min-cut is at most the average degree. 

\begin{lemma}\label{lem:uni-singleton-infeasibility-inequality}
Let $H=(V,E)$ be a loop-less hypergraph with a hyperedge-cost function $c:E\rightarrow \R_{\ge 0}$ and a budget $\lambda\in \R_{\ge 0}$. Let $F\subseteq E$ be a budget-feasible cut in $(H, c, \lambda)$. Let $I:=\{v\in V: c(\cut_H(\{v\})>\lambda\}$. If $I\neq \emptyset$, then 
\[
        \sum_{e\in E\setminus F}c(e)|I\cap e|
        >
        \sum_{e\in F}c(e)|I\setminus e|.
\]
\end{lemma}
\begin{proof}
By budget-feasibility of $F$ in $(H, c, \lambda)$, we have that 
\[
        c(\cut_H(\{v\}))>\lambda \ge c(F) \ \forall\ v\in I.
\]
Since $I\neq \emptyset$, summing over all vertices in $I$ gives
\[
        \sum_{v\in I}c(\cut_H(\{v\}))>|I|c(F).
\]
The input instance is loopless. 
Therefore,  
\[
\sum_{v\in I}c(\cut_H(\{v\}))
=\sum_{v\in I}\sum_{e\in E}c(e)1_{v\in e}
=\sum_{e\in E}c(e)\left(\sum_{v\in I}1_{v\in e}\right)
=\sum_{e\in E}c(e)|I\cap e|.
\] 
Hence, 
\[
        \sum_{e\in E}c(e)|I\cap e|>|I|\sum_{e\in F}c(e).
\]
Rearranging gives
\[
        \sum_{e\in E\setminus F}c(e)|I\cap e|
        >
        \sum_{e\in F}c(e)|I\setminus e|.
\]
\end{proof}

We obtain the following corollary by applying the above lemma to $(H, c_j, b_j)$. 
\begin{corollary}\label{coro:singleton-infeasibility-ineq}
Consider a loop-less hypergraph instance $(H=(V, E), \bar{c}, \bar{b})$ with cost functions $\bar{c}=(c_1\in \R^E_{\ge 0},\ldots,c_k\in \R^E_{\ge 0})$ and budgets $\bar{b}=(b_1,\ldots,b_k)\in \R^k_{\ge 0}$, such that $F\subseteq E$ is a budget-feasible cut in $(H, \bar{c}, \bar{b})$.  Let $j\in[k]$ and $I_j:=\{v\in V:c_j(\cut_H(\{v\}))>b_j\}$. If $I_j\neq \emptyset$, then 
\[
        \sum_{e\in E\setminus F}c_j(e)|I_j\cap e|
        >
        \sum_{e\in F}c_j(e)|I_j\setminus e|.
\]
\end{corollary}

\section{Quasi-Polynomial Time Algorithm}\label{sec:multi-quasi}
In this section, we design a quasi-polynomial time algorithm for $k$-objective hypergraph budgeted-cut and prove Theorem \ref{thm:k-objective-qp}.
For pedagogical reasons, we include the algorithm and the analysis of this section specialized for single-objective in \Cref{sec:uni}. That section also serves as a warm-up towards the algorithm in this section, since it introduces two of the three new ingredients of the algorithm in this section. We encourage readers unfamiliar with randomized algorithms for hypergraph min-cut to read that section. 

We will work with a generalization of the problem where the costs and budgets could be reals and not just integers. See Section \ref{sec:techniques} for an outline of the algorithm and  Algorithm 1 below for the pseudocode.

\begin{algorithm}[h]
\caption{QPT-Algo\((H=(V,E),\bar{c}=(c_1\in \R^E_{\ge 0},\ldots,c_k\in \R^E_{\ge 0}),\bar{b}=(b_1,\ldots, b_k)\in \R^k_{\ge 0})\)}
\begin{enumerate}[leftmargin=2em]
\item Remove all loops from \(H\).
\item Let $D:=\{e\in E:e=V\}$. Set
\begin{align*}
        H&\leftarrow H-D \text{ and}\\
        b_i&\leftarrow b_i-c_i(D)\quad\text{for every }i\in[k].
\end{align*}
If \(b_i<0\) for some $i\in[k]$, return \(\FAIL\).
\item If there exists \(v\in V\) such that \(c_i(\cut_H(\{v\}))\le b_i\) for every $i\in[k]$, return \(D\cup \cut_H(\{v\})\).
\item For every $i\in[k]$, let
\[
        I_i:=\{v\in V:c_i(\cut_H(\{v\}))>b_i\}.
\]
Choose $j\in[k]$ such that $|I_j|\ge |V|/k$.  Sample a hyperedge $e$ with probability
\[
        \Pr[e]=\frac{c_j(e)}{c_j(E)}.
\]
\item \(F_1\leftarrow \text{QPT-Algo}(H/e,\bar{c}|_{E-e},\bar{b})\).
\item If \(F_1\) is not \(\FAIL\), return \(F_1\cup D\).
\item With probability \(|e\cap I_j|/|I_j|\):
\begin{enumerate}[leftmargin=2em]
    \item $F_2'\leftarrow \text{QPT-Algo}(H\setminus e,\bar{c}|_{E-e},\bar{b}-\bar{c}(e))$.
    \item If \(F_2'\) is not \(\FAIL\), 
        find a connected component $U\subseteq V$ of $H-(F_2'\cup\{e\})$ and return $\delta_H(U)\cup D$.
\end{enumerate}
\item Return \(\FAIL\).
\end{enumerate}
\end{algorithm}

The sampling distribution in Step~4 is well-defined whenever that step is reached: Indeed, after Step~3 there is no budget-feasible singleton cut, so $\bigcup_{i=1}^{k} I_i=V$.  Hence, the chosen set $I_j$ is nonempty.  For every $v\in I_j$, we have $c_j(\cut_H(\{v\}))>b_j\ge 0$, and therefore some hyperedge has positive $c_j$-cost.  Hence, $c_j(E)>0$.

The algorithm deals with the recursive output from the contraction call and the deletion call slightly differently (Steps 6 and 7(b)). The output of the contraction call in Step 6 is returned directly while the output of the deletion call in Step 7(b) is refined before returning (both after  taking the union with $D$). The reason for refining in Step 7(b) is to ensure that the returned set is indeed a cut and not just the superset of a cut. 

\subsection{Correctness}

We first state the recursive preservation facts. These facts will help in showing that if the given instance has a budget-feasible cut, then the instances in the recursive calls have a budget-feasible cut, and conversely, if the recursive call returns a budget-feasible cut, then the algorithm indeed returns a budget-feasible cut.

\begin{proposition}\label{prop:multi-reductions}
Let $(H,\bar{c},\bar{b})$ be the input instance, where $\bar{c}=(c_1,\ldots,c_k)$ and $\bar{b}=(b_1,\ldots,b_k)$.
\begin{enumerate}[label=(\roman*),leftmargin=2em]
\item Deleting loops does not affect budget-feasibility.
\item Let $D$ be the set of spanning hyperedges.  If $(H,\bar{c},\bar{b})$ has a budget-feasible cut, then $b_i-c_i(D)\ge0$ for every $i\in[k]$ and the instance $(H-D,\bar{c}|_{E-D},\bar{b}-\bar{c}(D))$ has a budget-feasible cut.  Conversely, if $b_i-c_i(D)\ge0$ for every $i\in[k]$ and $F'$ is a budget-feasible cut in $(H-D,\bar{c}|_{E-D},\bar{b}-\bar{c}(D))$, then $F'\cup D$ is a budget-feasible cut in $(H,\bar{c},\bar{b})$.
\item For $e\in E$, if $F_1$ is a budget-feasible cut in $(H/e,\bar{c}|_{E-e},\bar{b})$, then $F_1$ is a budget-feasible cut in $(H,\bar{c},\bar{b})$.
\item For $e\in E$, given a budget-feasible cut $F_2'$ in $(H\setminus e,\bar{c}|_{E-e},\bar{b}-\bar{c}(e))$, 
if $U\subseteq V$ is a connected component of $H-(F_2'\cup\{e\})$, then 
$\delta_H(U)$ is a budget-feasible cut in $(H,\bar{c},\bar{b})$ and such a subset $U$ can be found in time $n^{O(1)}$. 
\end{enumerate}
\end{proposition}

\begin{proof}
We prove the statement below.
\begin{enumerate}[label=(\roman*),leftmargin=2em]
\item Loops do not cross any cut, so deleting them does not affect budget-feasibility.  
\item Every spanning hyperedge crosses every nontrivial cut and hence belongs to every cut.  Thus, if $F$ is a budget-feasible cut in $(H,\bar{c},\bar{b})$, then $F-D$ is a budget-feasible cut in the residual instance $(H-D, \bar{c}|_{E-D}, \bar{b}-\bar{c}(D))$ and the residual budgets are nonnegative.  Conversely, if $F'=\delta_{H-D}(U)$ is a budget-feasible cut in the residual instance $(H-D, \bar{c}|_{E-D}, \bar{b}-\bar{c}(D))$ for some non-empty proper subset $U\subseteq V$, then $F'\cup D=\delta_H(U)$ and has cost
\[
        c_i(F'\cup D)=c_i(F')+c_i(D)\le b_i
        \qquad\text{for every }i\in[k].
\]

\item Suppose that $F_1=\delta_{H/e}(U)$ is a budget-feasible cut in $(H/e,\bar{c}|_{E-e},\bar{b})$ for some non-empty proper subset $U\subseteq V(H/e)$. Without loss of generality, suppose that the contracted vertex corresponding to $e$ is not contained in $U$. 
Then, $F_1=\delta_H(U)$ and $c_i(F_1)\le b_i$ for every $i\in [k]$. 

\item 
By definition $\delta_H(U)$ is a cut with $\delta_H(U)\subseteq F_2'\cup \{e\}$ and hence, 
\[
        c_i(\delta_H(U))\le c_i(F_2'\cup\{e\})=c_i(F_2')+c_i(e)
        \le b_i-c_i(e)+c_i(e)=b_i
        \qquad\text{for every }i\in[k].
\]
The connected component $U\subseteq V$ of $H-(F_2'\cup\{e\})$ can be found in time $n^{O(1)}$. 
\end{enumerate}
\end{proof}

We now analyze the success probability. We briefly sketch the outline: Fix a budget-feasible cut $F=\cut_H(U)$ in the input instance.  If the sampled hyperedge $e$ is not in $F$, then $e$ is contained in $U$ or in $V\setminus U$, so contracting $e$ preserves the target cut and hence, budget-feasibility of the contraction instance.  If $e\in F$, then deleting $e$ and subtracting its cost vector preserves the residual target $F\setminus\{e\}$ and hence, budget-feasibility of the deletion instance.  The success-probability proof lower-bounds the probability that at least one of these two target-preserving events occurs at every recursive level.

\begin{lemma}[Success Probability]\label{lem:multi-success}
QPT-Algo returns either FAIL or a budget-feasible cut (i.e., it will not return a budget-infeasible cut). 
Moreover, if the input instance $(H,\bar{c},\bar{b})$ admits a budget-feasible cut and $|V|=n$, then one run of QPT-Algo returns a budget-feasible cut with probability at least
\[
        n^{-2k},
\]
where $k$ is the number of costs on hyperedges.
\end{lemma}

\begin{proof}
Suppose that the algorithm does not return FAIL. 
Every set returned by the algorithm is obtained from a budget-feasible cut returned by a recursive call, or is a singleton cut that is verified explicitly in Step~3.  Hence, by \Cref{prop:multi-reductions}, every set returned by the algorithm is a budget-feasible cut.  We next analyze the success probability.

Let $q_n^m$ be the least probability of returning a budget-feasible cut by QPT-Algo over all input instances with $n$ vertices and at most $m$ hyperedges that have a budget-feasible cut.  We prove that $q_n^m\ge n^{-2k}$ by induction on $n+m$.  The base case is $n+m=2$ and in particular $n=2$ and $m=0$. 
In this case, the algorithm will return a budget-feasible cut because the singleton cut is budget-feasible. 
We prove the induction step next. 

Fix an $n$-vertex $m$-hyperedge input instance $(H,\bar{c},\bar{b})$ that has a budget-feasible cut.  If it has a loop or a spanning hyperedge, then the claim follows using the induction hypothesis for $n+m-1$ and \Cref{prop:multi-reductions}.  Hence, we may assume that the input instance is loopless and has no spanning hyperedge.  If it has a vertex $v\in V$ such that $c_i(\cut_H(\{v\}))\le b_i$ for every $i\in[k]$, then the algorithm returns a budget-feasible cut with probability $1$.  Hence, we may assume that the input instance is loopless, has no spanning hyperedges, and has no budget-feasible singleton cut.  Thus, the algorithm reaches Step~4 and consequently, $n\ge 3$.

Since no singleton cut is budget-feasible, the singleton infeasible sets satisfy $\bigcup_{i=1}^{k}I_i=V$.  Let $j\in[k]$ be the index chosen in Step~4, so $|I_j|\ge n/k$.  
The sampling distribution in Step~4 is well-defined:  Indeed, after Step~3, there is no budget-feasible singleton cut, so $\bigcup_{i=1}^{k}I_i=V$.  The chosen set $I_j$ is nonempty.  For every $v\in I_j$, we have $c_j(\cut_H(\{v\}))>b_j\ge 0$, and therefore some hyperedge has positive $c_j$-cost.  Hence $c_j(E)>0$ and the algorithm indeed samples a hyperedge.
Since the instance is loopless and has no spanning hyperedges, every sampled hyperedge has size $a\in\{2,\ldots,n-1\}$.  

Let $F=\delta(U)$ be a budget-feasible cut in $(H,\bar{c},\bar{b})$ for some non-empty proper subset $U\subseteq V$.  Let $e\in E$ be the sampled hyperedge. 
We consider two cases.
\begin{enumerate}[leftmargin=2em]
\item Suppose $e\notin F$. Then, $(H/e,\bar{c}|_{E-e},\bar{b})$ has a budget-feasible cut:  Indeed, 
since $F=\delta(U)$ and $e\not\in \delta(U)$, we have that either $e\cap U=\emptyset$ or $e\subseteq U$. Without loss of generality, suppose $e\cap U=\emptyset$. Then, $F=\delta(U)$ is a budget-feasible cut in $(H/e, \bar{c}|_{E-e}, \bar{b})$.  The contraction instance has $n-|e|+1$ vertices and at most $m$ hyperedges.  Thus, the contraction call succeeds with probability at least $q^m_{n-|e|+1}$.
\item Suppose $e\in F$. Then, $(H\setminus e,\bar{c}|_{E-e},\bar{b}-\bar{c}(e))$ has a budget-feasible cut, namely $F\setminus\{e\}$.  The deletion instance has $n$ vertices and $m-1$ hyperedges.  The algorithm first tries the contraction call.  If it succeeds, then the algorithm is already successful.  If it fails, then the deletion branch is executed with probability $|e\cap I_j|/|I_j|$, and, conditioned on executing the deletion branch, succeeds with probability at least $q^{m-1}_n$. 
Consequently, the conditional success probability in this case is at least $(|e\cap I_j|/|I_j|)q^{m-1}_n$.
\end{enumerate}
Therefore,
\[
        q^m_n
        \ge
        \sum_{e\in E\setminus F}\left(\frac{c_j(e)}{c_j(E)}\right)q^m_{n-|e|+1}
        +
        \sum_{e\in F}\left(\frac{c_j(e)}{c_j(E)}\right)\left(\frac{|e\cap I_j|}{|I_j|}\right)q^{m-1}_n.
\]
By the induction hypothesis, $q^m_{n-|e|+1}\ge (n-|e|+1)^{-2k}$ and $q^{m-1}_n\ge n^{-2k}$.  
Substituting gives
\begin{align*}
        q^m_n
        &\ge
        \sum_{e\in E\setminus F}\left(\frac{c_j(e)}{c_j(E)}\right)(n-|e|+1)^{-2k}
        +
        \sum_{e\in F}\left(\frac{c_j(e)}{c_j(E)}\right)\frac{|e\cap I_j|}{|I_j|}n^{-2k}\\
        &=
        n^{-2k}\left(
        \sum_{e\in E\setminus F}\left(\frac{c_j(e)}{c_j(E)}\right)\left(\frac{n}{n-|e|+1}\right)^{2k}
        +
        \sum_{e\in F}\left(\frac{c_j(e)}{c_j(E)}\right)\frac{|e\cap I_j|}{|I_j|}
        \right)\\
        &=
        n^{-2k}\left(
        1+
        \sum_{e\in E\setminus F}\left(\frac{c_j(e)}{c_j(E)}\right)\left(\left(\frac{n}{n-|e|+1}\right)^{2k}-1\right)
        -
        \sum_{e\in F}\left(\frac{c_j(e)}{c_j(E)}\right)\left(1-\frac{|e\cap I_j|}{|I_j|}\right)
        \right).
\end{align*}
For every hyperedge $e\in E$, by \Cref{prop:bi-numeric} (applied for $n$, $a=|e|$ and $t=2k$) and the fact that $|I_j|\ge n/k$, we have that 
\[
        \left(\frac{n}{n-|e|+1}\right)^{2k}-1
        \ge \frac{k|e|}{n}
        \ge \frac{|e|}{|I_j|}
        \ge \frac{|e\cap I_j|}{|I_j|}. 
\]
Thus, 
\[
q^m_n\ge n^{-2k}\left(1+\sum_{e\in E\setminus F}\left(\frac{c_j(e)}{c_j(E)}\right)\left(\frac{|e\cap I_j|}{|I_j|}\right)
        -
        \sum_{e\in F}\left(\frac{c_j(e)}{c_j(E)}\right)\left(\frac{|I_j\setminus e|}{|I_j|}\right)\right).
\]
Hence, it suffices to show that
\[
        \sum_{e\in E\setminus F}\left(\frac{c_j(e)}{c_j(E)}\right)\left(\frac{|e\cap I_j|}{|I_j|}\right)
        \ge
        \sum_{e\in F}\left(\frac{c_j(e)}{c_j(E)}\right)\left(\frac{|I_j\setminus e|}{|I_j|}\right), 
\]
which holds by \Cref{coro:singleton-infeasibility-ineq} since $I_j\neq \emptyset$.
\end{proof}

\subsection{Runtime}
We first bound the expected number of leaves in the recursion tree of the  algorithm. 

\begin{lemma}
\label{lem:multi-leaves}
The expected number of leaves in the recursion tree of QPT-Algo on an input instance with $n$ vertices and $m$ hyperedges is at most 
\[
(m+1)^{4k\log_2{n}}. 
\]
\end{lemma}
\begin{proof}
Let \(L(n,m)\) be the worst-case expected number of leaves in the recursion tree of QPT-Algo on an input instance with \(n\) vertices and \(m\) hyperedges (irrespective of whether the input instance has a budget-feasible cut or not). 
The summary is that we have the following recurrence for $L(n, m)$:
\[
        L(n,m)\lesssim L(n-|e|+1,m-1)+\rho L(n,m-1),
        \qquad
        \text{ where } \rho=\frac{|e\cap I_j|}{|I_j|}\le \min\{1,k|e|/n\}.
\]
For graphs, $|e|=2$ and $\rho=O(1/n)$ and the bound is indeed polynomial for fixed constant $k$; the polynomial bound extends to constant-rank hypergraphs; for arbitrary rank hypergraphs, $\rho$ can be $1$ since  $e\cap I_j$ could be equal to $I_j$.  This is exactly why we have a bound of the form $m^{O(k\log n)}$ rather than a polynomial in $n$. We now show the details.

We will prove by induction on $n+m$ that $L(n,m)\le (m+1)^{4k\log_2{n}}$. 
The base case corresponds to $n+m=2$ and in particular $n=2$ and $m=0$: the algorithm terminates without any recursive calls and hence $L(n,m)=1$. We now show the induction step. 

Let $(H, \bar{c}, \bar{b})$ be an input instance with $n$ vertices and $m$ hyperedges. If $n=2$, then by steps 1, 2, and 3, the algorithm will terminate without any recursive calls and hence $L(n, m)=1$. Hence, we may assume that $n\ge 3$. If there exists a singleton cut that is budget-feasible, then $L(n, m)\le 1$. Hence, we may assume that no singleton cut is budget-feasible. 
Let $j\in[k]$ such that $|I_j|\ge n/k$. Let $e$ be the hyperedge sampled in Step 4. We note that $|e|\ge 2$ since there are no loops when we reach Step 4. 
The contraction branch instance has \(n-|e|+1\) vertices because \(|e|\) vertices are identified into one.  The deletion branch instance has the same number of vertices, but one fewer hyperedge. Hence,
\begin{align*}
    &E\left[\text{Number of leaves in the recursion tree for }(H, \bar{c}, \bar{b})|e\text{ is sampled}\right]\\
    &\quad \quad \le L(n-|e|+1,m-1) + \frac{|e\cap I_j|}{|I_j|}L(n, m-1)\\
    &\quad \quad \le m^{4k\log{(n-|e|+1})} + \frac{|e\cap I_j|}{|I_j|}m^{4k\log{n}} \quad \quad \text{(by induction hypothesis)}\\
    &\quad \quad \le (m+1)^{4k\log{n}}\left(\frac{1}{(m+1)^{4k\log{\left(\frac{n}{n-|e|+1}\right)}}} + \frac{|e\cap I_j|}{|I_j|}\left(1-\frac{1}{m+1}\right)^{4k\log{n}}\right). 
\end{align*}
We will now show that 
\[
\frac{1}{(m+1)^{4k\log{\left(\frac{n}{n-|e|+1}\right)}}} + \frac{|e\cap I_j|}{|I_j|}\left(1-\frac{1}{m+1}\right)^{4k\log{n}}\le 1. 
\]
We recall that $|I_j|\ge n/k$ and $|e|-1\ge |e\cap I_j|/2$ (since $|e|\ge 2$) and hence, 
\begin{align}
    |e|-1\ge \frac{1}{2}\left(\frac{|e\cap I_j|}{|I_j|}\right)|I_j| \ge \frac{n}{2k}\left(\frac{|e\cap I_j|}{|I_j|}\right). \label{ineq:multi-large-e}
\end{align}
We consider two cases. 

\begin{enumerate}
    \item Suppose $|e\cap I_j|/|I_j|\le 1/2$. 
    Using inequality (\ref{ineq:multi-large-e}), we have that 
    \[
    \log{\left(\frac{n}{n-|e|+1}\right)} = -\log\left({1-\frac{|e|-1}{n}}\right) \ge \frac{|e|-1}{n} \ge \frac{1}{2k}\left(\frac{|e\cap I_j|}{|I_j|}\right). 
    \]
    Hence, 
    \begin{align*}
    \frac{1}{(m+1)^{4k\log{\left(\frac{n}{n-|e|+1}\right)}}} + \frac{|e\cap I_j|}{|I_j|} \left(1-\frac{1}{m+1}\right)^{4k\log{n}}
    &\le \frac{1}{(m+1)^{2\left(\frac{|e\cap I_j|}{|I_j|}\right)}} + \frac{|e\cap I_j|}{|I_j|}\\
    &\le \frac{1}{2^{2\left(\frac{|e\cap I_j|}{|I_j|}\right)}} + \frac{|e\cap I_j|}{|I_j|} \\
    &\le 1-\frac{|e\cap I_j|}{|I_j|} + \frac{|e\cap I_j|}{|I_j|}= 1. 
    \end{align*}
    The last inequality above is because $2^{-2x}\le 1-x$ for $x\le 1/2$, where $x:=|e\cap I_j|/|I_j|$. 
    \item Suppose $|e\cap I_j|/|I_j|>1/2$. 
    Using inequality (\ref{ineq:multi-large-e}), we have that $|e|-1>n/(4k)$. Consequently, 
    \[
    \log{\left(\frac{n}{n-|e|+1}\right)} > \log{\left(\frac{n}{n-n/(4k)}\right)} = \log{\frac{4k}{4k-1}}. 
    \]
    Hence, 
    \begin{align*}
    &\frac{1}{(m+1)^{4k\log{\left(\frac{n}{n-|e|+1}\right)}}} + \frac{|e\cap I_j|}{|I_j|} \left(1-\frac{1}{m+1}\right)^{4k\log{n}}\\
    &\quad \quad \quad \quad \quad \quad \quad \quad < \frac{1}{(m+1)^{4k\log{\frac{4k}{4k-1}}}} + \left(1-\frac{1}{m+1}\right)^{4k\log{n}}\\
    &\quad \quad \quad \quad \quad \quad \quad \quad \le \frac{1}{m+1} + \left(1-\frac{1}{m+1}\right)^{4k} \quad \quad \text{(since $n\ge 2$)}\\
    &\quad \quad \quad \quad \quad \quad \quad \quad \le 1. 
    \end{align*}
\end{enumerate}
\end{proof}

We now bound the number of nodes in the recursion tree of the multiobjective algorithm. 
\begin{lemma}\label{lem:multi-runtime}
The expected number of nodes in the recursion tree of QPT-Algo on an input instance with $n$ vertices and $m$ hyperedges is at most 
\[
        (m+1)^{5k\log_2 n}.
\]
\end{lemma}
\begin{proof}
    Every root to leaf path in the recursion tree has at most $m$ hyperedges since each recursive call reduces the number of hyperedges by at least $1$. Therefore, the total number of nodes in the recursion tree is at most $m+1$ times the number of leaves in the recursion tree. Hence, by \Cref{lem:multi-leaves}, the expected number of nodes in the recursion tree is at most $(m+1)\cdot (m+1)^{4k\log{n}}\le (m+1)^{5k\log{n}}$. 
\end{proof}

All steps of QPT-Algo can be implemented to run in polynomial time---the only non-trivial step is Step 7(b), which can be implemented to run in polynomial time by Proposition \ref{prop:multi-reductions}(iv). Thus, \Cref{lem:multi-success,lem:multi-runtime} together prove \Cref{thm:k-objective-qp}. 

\section{PTAS}\label{sec:multi-ptas}

In this section, we design a PTAS for $k$-objective hypergraph budgeted-cut and prove \Cref{thm:k-ptas}. We will work with a generalization of the problem where the costs and budgets could be reals and not just integers.

\paragraph{Algorithm Outline.}
Let the input instance be $(H=(V,E),\bar{c}=(c_1 \in \R^E_{\ge 0},\ldots, c_k \in \R^E_{\ge 0}),\bar{b}\in \R^k_{\ge 0})$. We define $\bar{b}':=(b'_1, \ldots, b'_k) \in \R^k_{\ge 0}$, where $b'_i:=(1+\varepsilon)\cdot b_i$ for every $i\in [k]$. The goal is to find a budget-feasible cut in $(H, \bar{c}, \bar{b}')$. If the hypergraph $H$ contains at most $8k$ vertices, i.e., $|V|\leq 8k$, the algorithm verifies budget-feasibility of all cuts of $H$ by brute-force and returns a budget-feasible one if one exists and otherwise, returns FAIL.

We may now assume that $|V|>8k$. The algorithm begins with deterministic reductions.  First, it removes all loops by noting that loops are not part of any cuts.  Then, it collects all spanning hyperedges into a set $D$ with the observation being that spanning hyperedges are part of every cut; the rest of the algorithm works with the residual hypergraph $H=(V, E-D)$ with the same hyperedge costs and residual budgets $\bar{b}'-\bar{c}(D)$. 
The algorithm checks if there exists a singleton cut $\delta_H(v)$ that is budget-feasible in $(H, \bar{c}, \bar{b}')$ and if so, then returns $\delta_H(v)\cup D$.

Next, the algorithm verifies whether there exists a hyperedge $e\in E$ with $|e|\geq |V|/(2k)$. Suppose such a hyperedge does not exist.
Since no singleton cut is budget-feasible, for every vertex $v\in V$, there exists $i\in[k]$ such that $c_i(\cut_H(\{v\}))>b'_i$. 
For each $i\in [k]$, define the singleton-infeasible set
\[
        I_i:=\{v\in V:c_i(\cut_H(\{v\}))>b'_i\}.
\]
Then $\bigcup_{i=1}^{k}I_i=V$, so one of these $k$ singleton-infeasible sets has size at least $|V|/k$.  The algorithm chooses an index $j\in[k]$ with
\[
        |I_j|\ge \frac{|V|}{k},
\]
samples a hyperedge $e$ with probability proportional to its $c_j$-cost. Next, the algorithm recurses on the contracted instance \(H/e\) with the same costs on surviving hyperedges and with the same budget $\bar{b}'$, termed as a \emph{small-contraction call}. If the small-contraction call succeeds with a budget-feasible cut $F_1$, then the algorithm returns $F_1\cup D$.

Next, suppose that there exists a hyperedge containing at least $|V|/(2k)$ vertices. Define $L:=\left\{e\in E: |e|\geq \frac{|V|}{4k}\right\}$.
The algorithm samples integer $X\in\{0, 1, \ldots, k\}$ uniformly at random and branches to one of the following:
\begin{enumerate}
    \item $X=0$: The algorithm recurses on the deletion instance $H-L$ with the same costs on the surviving hyperedges and with residual budgets $\bar{b}'-\bar{c}(L)$, termed as a \emph{deletion call}. If the deletion call succeeds with a budget-feasible cut $F_2$, then the algorithm returns a budget-feasible cut contained within $F_2\cup L\cup D$ (see \Cref{prop:ptas-multi-reductions}(iv) for a proof of existence of a budget-feasible cut within $F_2\cup L\cup D$). Else, the algorithm returns FAIL.
    \item $X>0$: The algorithm samples a hyperedge $e\in L$ with probability proportional to its $c_X$-cost. Next, the algorithm recurses on the contracted instance \(H/e\) with the same costs on surviving hyperedges and with the same budgets $\bar{b}'$, termed as a \emph{large-contraction call}. If the large-contraction call succeeds with a budget-feasible cut $F_2$, then the algorithm returns $F_2\cup D$.
\end{enumerate}
See PTAS-Algo below for a pseudocode.
\begin{algorithm}[!t]
\caption{PTAS-Algo\((H=(V,E),\bar{c}=(c_1\in \R^E_{\ge 0},\ldots,c_k\in \R^E_{\ge 0}),\bar{b}'=(b'_1,\ldots,b'_k)\in \R^k_{\ge 0})\)}
\begin{enumerate}[leftmargin=2em]
\item If $|V|\leq 8k$, brute-force to return a budget-feasible cut if one exists, else return FAIL.
\item Remove all loops from \(H\).
\item Let $D:=\{e\in E:e=V\}$. Set
\begin{align*}
        H&\leftarrow H-D \text{ and}\\
        b'_i&\leftarrow b'_i-c_i(D)\quad\text{for every }i\in[k].
\end{align*}
If \(b'_i<0\) for some $i\in[k]$, return \(\FAIL\).
\item If there exists \(v\in V\) such that \(c_i(\cut_H(\{v\}))\le b'_i\) for every $i\in[k]$, return \(D\cup \cut_H(\{v\})\).
\item If $|e|< |V|/(2k)$ for every $e\in E$, then:
\begin{enumerate}[leftmargin=2em]
    \item For every $i\in[k]$, let
    \[
        I_i:=\{v\in V:c_i(\cut_H(\{v\}))>b'_i\}.
    \]
    Choose $j\in[k]$ such that $|I_j|\ge |V|/k$.  Sample a hyperedge $e$ with probability
    \[
        \Pr[e]=\frac{c_j(e)}{c_j(E)}.
    \]
    \item \(F_1\leftarrow \text{PTAS-Algo}(H/e,\bar{c}|_{E-e},\bar{b}')\).
    \item If \(F_1\) is not \(\FAIL\), return \(F_1\cup D\).
\end{enumerate}
\item Else: //Randomly pick a branch
\begin{enumerate}[leftmargin=2em]
    \item $L\leftarrow \{e\in E:|e|\geq |V|/(4k)\}$.
    \item Sample integer $X\in \{0,1, \ldots, k\}$ uniformly at random.
    \item If $X>0$ and $c_X(L)=0$, $X\leftarrow 0$.
    \item If $X=0$:
    \begin{enumerate}[leftmargin=2em]
        \item If there exists $i\in[k]$ with $b'_i<c_i(L)$, return FAIL.
        \item \(F_2\leftarrow \text{PTAS-Algo}(H-L,\bar{c}|_{E-L},\bar{b}'-\bar{c}(L))\).
        \item If \(F_2\) is not \(\FAIL\), 
        find a connected component $U$ of $H-(F_2\cup L)$ and return $\delta_H(U)\cup D$.
    \end{enumerate}
    \item Else:
    \begin{enumerate}[leftmargin=2em]
        \item Sample a hyperedge $e$ with probability
        \[
            \Pr[e]=\frac{c_X(e)}{c_X(L)}.
        \]
        \item \(F_2\leftarrow \text{PTAS-Algo}(H/e,\bar{c}|_{E-e},\bar{b}')\).
        \item If \(F_2\) is not \(\FAIL\), return \(F_2\cup D\).
    \end{enumerate}
\end{enumerate}
\item Return \(\FAIL\).
\end{enumerate}
\end{algorithm}

The sampling distribution in Step~5(a) is well-defined whenever that step is reached. Indeed, after Step~4 there is no budget-feasible singleton cut, so $\bigcup_{i=1}^{k}I_i=V$. The chosen set $I_j$ is nonempty. For every $v\in I_j$, we have $c_j(\cut_H(\{v\}))>b_j'\ge 0$, and therefore some hyperedge has positive $c_j$-cost. Hence $c_j(E)>0$. Moreover, the sampling distribution is Step~6(e) is also well-defined whenever that step is reached. This is because if Step 6(e) is reached, then algorithm has already executed Step 6(c) which implies that $c_X(L)>0$.

The algorithm deals with the recursive output from the contraction calls and the deletion call slightly differently (Steps 5(b) and 6(e)(iii) versus Step 6(d)(iii)). The output of contraction calls in Steps 5(b) and 6(e)(iii) are returned directly while the output of the deletion call in Step 6(d)(iii) is refined before returning (all after taking the union with $D$). The reason for refining in Step 6(d)(iii) is to ensure that the returned set is indeed a cut and not just the superset of a cut.

\subsection{Correctness and Runtime Analysis}

Let $\varepsilon>0$. 
We say that an instance $(H,\bar{c},\bar{b}')$ is $(1+\varepsilon)$-strongly budget-feasible if instance $(H,\bar{c},\frac{1}{1+\varepsilon}\cdot \bar{b}')$ is budget-feasible. 
We say a cut $F$ of $H$ is a $(1+\varepsilon)$-strongly budget-feasible cut of $(H,\bar{c},\bar{b}')$ if $c_i(F)\leq \left(\frac{1}{1+\varepsilon}\right)b'_i$ for every $i \in [k]$. We note that a $(1+\varepsilon)$-strongly budget-feasible cut of $(H,\bar{c},\bar{b}')$ is also a budget-feasible cut of $(H, \bar{c}, \bar{b}')$. 
We first state the recursive preservation facts. These facts will help in showing that if the given instance has a $(1+\varepsilon)$-strongly budget-feasible cut, then the instances in the recursive calls have a $(1+\varepsilon)$-strongly budget-feasible cut, and conversely, if the recursive call returns a budget-feasible cut, then the algorithm indeed returns a budget-feasible cut. 

\begin{proposition}\label{prop:ptas-multi-reductions}
Let $(H,\bar{c},\bar{b}')$ be the input instance, where $\bar{c}=(c_1,\ldots,c_k)$ and $\bar{b}'=(b'_1,\ldots,b'_k)$. Let $\varepsilon>0$ be a positive parameter.
\begin{enumerate}[label=(\roman*),leftmargin=2em]
\item Deleting loops does not affect $(1+\varepsilon)$-strongly budget-feasibility.
\item Let $D$ be the set of spanning hyperedges.  If $(H,\bar{c},\bar{b}')$ is $(1+\varepsilon)$-strongly budget-feasible, then $b'_i-c_i(D)\ge0$ for every $i\in[k]$ and the instance $(H-D,\bar{c}|_{E-D},\bar{b}'-\bar{c}(D))$ is $(1+\varepsilon)$-strongly budget-feasible. Conversely, if $b'_i-c_i(D)\ge0$ for every $i\in[k]$ and $F'$ is a budget-feasible cut in $(H-D,\bar{c}|_{E-D},\bar{b}'-\bar{c}(D))$, then $F'\cup D$ is a budget-feasible cut in $(H,\bar{c},\bar{b}')$.
\item For $e\in E$, if $F_1$ is a $(1+\varepsilon)$-strongly budget-feasible cut in $(H/e,\bar{c}|_{E-e},\bar{b}')$, then $F_1$ is a $(1+\varepsilon)$-strongly budget-feasible cut in $(H,\bar{c},\bar{b}')$.
\item For $A\subseteq E$, given a budget-feasible cut $F_2$ in $(H-A, \bar{c}|_{E-A}, \bar{b}'-\bar{c}(A))$, 
if $U\subseteq V$ is a connected component of $H-(F_2\cup A)$, then 
$\delta_H(U)$ is a budget-feasible cut in $(H,\bar{c},\bar{b}')$ and such a subset $U$ can be found in time $n^{O(1)}$.
\end{enumerate}
\end{proposition}

\begin{proof}
We prove the statement below.
\begin{enumerate}[label=(\roman*),leftmargin=2em]
\item Loops cross no cut, so deleting them does not affect $(1+\varepsilon)$-strongly budget-feasibility. 
\item Every spanning hyperedge crosses every nontrivial cut and hence belongs to every cut. We observe that if $F$ is a $(1+\varepsilon)$-strongly budget-feasible cut in $(H,\bar{c},\bar{b}')$, then
\[
    c_i(F-D) = c_i(F)-c_i(D)\leq \frac{1}{1+\varepsilon}\cdot b'_i-c_i(D)\leq \frac{1}{1+\varepsilon}\cdot (b'_i-c_i(D))
    \qquad\text{for every }i\in[k],
\]
which implies that $F-D$ is a $(1+\varepsilon)$-strongly budget-feasible cut in the residual instance $(H-D, \bar{c}|_{E-D}, \bar{b}'-\bar{c}(D))$ and the residual budgets are nonnegative. Conversely, if $F'=\delta_{H-D}(U)$ is a budget-feasible cut in the residual instance  $(H-D, \bar{c}|_{E-D}, \bar{b}'-\bar{c}(D))$ for some non-empty proper subset $U\subseteq V$, then $F'\cup D=\delta_H(U)$ and has cost
\[
        c_i(F'\cup D)=c_i(F')+c_i(D)\le b'_i
        \qquad\text{for every }i\in[k].
\]

\item Suppose that $F_1=\delta_{H/e}(U)$ is a $(1+\varepsilon)$-strongly budget-feasible cut in $(H/e,\bar{c}|_{E-e},\bar{b}')$ for some non-empty proper subset $U\subseteq V(H/e)$.  Without loss of generality, suppose that the contracted vertex corresponding to $e$ is not contained in $U$. Then, $F_1=\delta_H(U)$ and $c_i(F_1)\le \frac{1}{1+\varepsilon}\cdot b_i'$ for every $i\in [k]$.

\item 
By definition $\delta_H(U)$ is a cut with $\delta_H(U)\subseteq F_2\cup A$ and hence, 
\[
        c_i(\delta_H(U))\le c_i(F_2\cup A)=c_i(F_2)+c_i(A)
        \le b'_i-c_i(A)+c_i(A)=b'_i
        \qquad\text{for every }i\in[k].
\]
The connected component $U\subseteq V$ of $H-(F_2\cup A)$ can be found in time $n^{O(1)}$. 
\end{enumerate}
\end{proof}

We will denote Step 6 as a branching step. We now bound the number of branching steps performed by the algorithm. 

\begin{lemma}\label{lemma:multi-fptas-depth}
    The total number of branching steps in one execution of PTAS-Algo on a $n$-vertex hypergraph is at most $8k\cdot \log{n}$.
\end{lemma}
\begin{proof}
    We prove by induction on $n$. For the base case of $n\leq 8k$, PTAS-Algo terminates without branching due to the brute-force step. 
    We now prove the induction step. Suppose $|e|<|V|/(2k)$ for every hyperedge $e\in E$. Then, PTAS-Algo will execute Step~5, where it contracts a hyperedge $e$ and recurses on the instance $(H/e, \bar{c}|_{E-e}, \bar{b}')$, where hypergraph $H/e$ contains at most $n-1$ vertices. By induction hypothesis, the execution of PTAS-Algo on instance $(H/e,\bar{c}|_{E-e},\bar{b}')$ performs at most $8k\cdot \log (n-1)$ branching steps. Hence, the execution of PTAS-Algo on instance $(H,\bar{c},\bar{b}')$ performs at most $8k\cdot \log (n-1)<8k\cdot \log n$ branching steps. Next, suppose there exists a hyperedge $e\in E$ with $|e|\geq |V|/(2k)$. Then, PTAS-Algo will execute Step~6 and branch to $X=0$ or $X>0$. We consider the following two cases:
    \begin{enumerate}
        \item $X=0$: PTAS-Algo executes the if-clause in Step~6(d) and recurses on the instance $(H-L, \bar{c}|_{E-L},\bar{b}'-\bar{c}(L))$. We observe that every remaining hyperedge $e\in E\setminus L$ has $|e|<|V|/(4k)$. This implies that the next branching step can be performed only when the number of vertices has decreased to $n'<n/2$. According to the induction hypothesis, the total number of branching steps performed is at most
        \[
            1+8k\cdot \log n'<1+8k\cdot \log (n/2)<8k\cdot \log n.
        \]

        \item $X>0$: PTAS-Algo executes the else-clause in Step~6(e), where it contracts a hyperedge $e\in L$ and recurses on the instance $(H/e,\bar{c}|_{E-e},\bar{b}')$. Since $e\in L$ and $|e|\geq n/(4k)$, the number of vertices in hypergraph $H/e$ is at most $n'=n-|e|+1\leq \frac{4k-1}{4k}\cdot n+1$. Using the induction hypothesis, the total number of branching steps performed is at most
        \[
            1+8k\cdot \log n'\leq 1+8k\cdot \log \left(\frac{4k-1}{4k}\cdot n+1\right)<8k\cdot \log n,
        \]
        where the last inequality holds since $n> 8k$.
    \end{enumerate}
\end{proof}

We now analyze the success probability.

\begin{lemma}[Success probability]\label{lem:multi-ptas-success}
Let $(H,\bar{c},\bar{b})$ be the input instance and $b'_i:=(1+\varepsilon)\cdot b_i$ for every $i\in [k]$. Let $n:=|V|$. 
PTAS-Algo on input $(H, \bar{c}, \bar{b}')$ returns either FAIL or a budget-feasible cut in $(H,\bar{c},\bar{b}')$ (i.e., it will not return a budget-infeasible cut).
Moreover, if the input instance $(H, \bar{c}, \bar{b})$ admits a budget-feasible cut, then one run of PTAS-Algo on instance $(H,\bar{c},\bar{b}')$ returns a budget-feasible cut of $(H, \bar{c}, \bar{b}')$ with probability at least
\[
        n^{-4k}\cdot \left(\frac{\varepsilon}{(k+1)\cdot(1+\varepsilon)}\right)^{8k\cdot \log n}. 
\]
\end{lemma}

\begin{proof}
Suppose that the algorithm does not return FAIL. 
Every returned set is obtained from a budget-feasible cut returned by a recursive call, or is a singleton cut that is verified explicitly in Step~3.  Hence, by Proposition \ref{prop:ptas-multi-reductions}, every set returned by the algorithm is a budget-feasible cut in $(H, \bar{c},\bar{b}')$.  We next analyze the success probability.

Let $q_{n,m}^{\ell}$ be the least probability of returning a budget-feasible cut of PTAS-Algo over all input instances $(H, \bar{c}, \bar{b}')$ with $n$ vertices and $m$ hyperedges that have a $(1+\varepsilon)$-strongly budget-feasible cut on which PTAS-Algo performs at most $\ell$ branching steps.  We prove that 
\begin{align}\label{inequality:multi-ptas-success-probability}
    q_{n,m}^{\ell}\geq n^{-4k}\cdot \left(\frac{\varepsilon}{(k+1)\cdot (1+\varepsilon)}\right)^{\ell}
\end{align}
by induction on $n+m+\ell$. We fix an input instance $(H,\bar{c},\bar{b}')$ with $n$ vertices and $m$ hyperedges that has a $(1+\varepsilon)$-strongly budget-feasible cut such that the algorithm performs at most $\ell$ branching steps.

For the base case, we consider $n+m+\ell\leq 8k$. In particular, $n\le 8k$ and thus, PTAS-Algo finds all cuts of $H$ by brute-force and returns a budget-feasible cut. Consequently, the algorithm indeed returns a budget-feasible cut in $(H, \bar{c}, \bar{b}')$ with probability $1$.

We now prove the induction step for $n+m+\ell>8k$. If $n\leq 8k$, PTAS-Algo would find all cuts of $H$ by brute-force and return a budget-feasible cut in $(H, \bar{c}, \bar{b}')$ with probability $1$. We now consider $n>8k$.
If $H$ has a loop or a spanning hyperedge, then the claim follows using the induction hypothesis for $n+(m-1)+\ell$ and Proposition \ref{prop:ptas-multi-reductions}. 
Hence, we may assume that the input instance is loopless and has no spanning hyperedges. Thus, the algorithm reaches Step~4. Let $F=\delta_H(U)$ be a $(1+\varepsilon)$-strongly budget-feasible cut in $(H,\bar{c},\bar{b}')$ for some non-empty proper subset $U\subseteq V$, where $c_i(F)\leq \frac{1}{1+\varepsilon}\cdot b'_i$ for every $i\in [k]$. 
We consider the following three cases: 
\begin{enumerate}
    \item $|e|<|V|/(2k)$ for every $e\in E$: If there exists a vertex $v\in V$ such that $c_i(\cut_H(\{v\}))\le b'_i$ for every $i\in[k]$, then the algorithm returns a budget-feasible cut of $(H, \bar{c}, \bar{b}')$ with probability $1$. Otherwise, the singleton infeasible sets satisfy $\bigcup_{i=1}^{k}I_i=V$.  Let $j\in[k]$ be the index chosen in Step~(b), so $|I_j|\ge n/k$. The sampling distribution in Step~(b) is well-defined: Indeed, after Step~(a) there is no budget-feasible singleton cut, so $\bigcup_{i=1}^{k}I_i=V$. The chosen set $I_j$ is nonempty.  For every $v\in I_j$, we have $c_j(\cut_H(\{v\}))>b_j'\ge 0$, and therefore some hyperedge has positive $c_j$-cost. Hence $c_j(E)>0$.

    Let $e\in E$ be the sampled hyperedge. Suppose $e\notin F$. Then, $|e|\in \{2, 3, \ldots, n-1\}$. Moreover, $(H/e,\bar{c}|_{E-e},\bar{b}')$ has a $(1+\epsilon)$-strongly budget-feasible cut: Indeed, since $F=\delta_H(U)$ and $e\not\in \delta_H(U)$, we have that either $e\cap U=\emptyset$ or $e\subseteq U$. Without loss of generality, suppose $e\cap U=\emptyset$. Then, $F=\delta_{H/e}(U)$ is a $(1+\varepsilon)$-strongly budget-feasible cut in $(H/e, \bar{c}|_{E-e}, \bar{b}')$. 
    The contraction instance has $n-|e|+1$ vertices. Thus, the contraction call succeeds with probability at least $q^{\ell}_{n-|e|+1,m-1}$.

    Therefore,
    \[
        q^{\ell}_{n,m}
        \ge
        \sum_{e\in E\setminus F}\left(\frac{c_j(e)}{c_j(E)}\right)q^{\ell}_{n-|e|+1,m-1}.
    \]
    By the induction hypothesis, $q^{\ell}_{n-|e|+1,m-1}\ge (n-|e|+1)^{-4k}\cdot (\frac{\varepsilon}{(k+1)\cdot (1+\varepsilon)})^{\ell}$.  Let $p_e:=c_j(e)/c_j(E)$.  Substituting gives
    \begin{align*}
        q^{\ell}_{n,m}
        &\ge
        \sum_{e\in E\setminus F}p_e(n-|e|+1)^{-4k}\cdot \left(\frac{\varepsilon}{(k+1)\cdot (1+\varepsilon)}\right)^{\ell}\\
        &=
        n^{-4k}\cdot \left(\frac{\varepsilon}{(k+1)\cdot(1+\varepsilon)}\right)^{\ell}\cdot \left(
        \sum_{e\in E\setminus F}p_e\left(\frac{n}{n-|e|+1}\right)^{4k}
        \right)\\
        &=
        n^{-4k}\cdot \left(\frac{\varepsilon}{(k+1)\cdot(1+\varepsilon)}\right)^{\ell}\cdot\left(
        1+
        \sum_{e\in E\setminus F}p_e\left(\left(\frac{n}{n-|e|+1}\right)^{4k}-1\right)
        -
        \sum_{e\in F}p_e
        \right).
    \end{align*}
    By \Cref{prop:bi-numeric} (applied for $n$, $a=|e|$ and $t=4k$) and the fact that $|I_j|\ge n/k$, we have that 
    \[
        \left(\frac{n}{n-|e|+1}\right)^{4k}-1
        \ge \frac{2k|e|}{n}
        \ge \frac{2|e|}{|I_j|}
        \ge \frac{2|e\cap I_j|}{|I_j|}
    \]
    for every hyperedge $e$.  Thus, 
    \[
    q^{\ell}_{n,m}\ge n^{-4k}\cdot \left(\frac{\varepsilon}{(k+1)\cdot(1+\varepsilon)}\right)^{\ell}\cdot \left(1+2\sum_{e\in E\setminus F}p_e\left(\frac{|e\cap I_j|}{|I_j|}\right)
        -
        \sum_{e\in F}p_e\right).
    \]
    Hence, it suffices to show that
    \[
        2\sum_{e\in E\setminus F}p_e\left(\frac{|e\cap I_j|}{|I_j|}\right)
        \ge
        \sum_{e\in F}p_e.
    \]
    Substituting $p_e=c_j(e)/c_j(E)$, it suffices to show that 
    \[
    2\sum_{e\in E\setminus F}c_j(e)|e\cap I_j|\ge \sum_{e\in F}c_j(e)|I_j|,
    \]
    which is correct since by \Cref{coro:singleton-infeasibility-ineq} applied to $I_j$ (recall that $I_j\neq \emptyset$), we have
    \[
        \sum_{e\in E\setminus F}c_j(e)|I_j\cap e|
        >
        \sum_{e\in F}c_j(e)(|I_j\setminus e|)
        >
        \sum_{e\in F}c_j(e)(|I_j|-|e|)
        >
        \sum_{e\in F}c_j(e)\cdot \frac{|I_j|}{2},
    \]
    where the last inequality is because $|e|<|V|/(2k) \leq |I_j|/2$ for all $e\in E$.

    \item There exists a hyperedge $e\in E$ with $|e|\geq |V|/(2k)$ and $c_i(L\setminus F)\leq \frac{\varepsilon}{1+\varepsilon}\cdot c_i(L)$ for every $i\in [k]$: In this case, PTAS-Algo branches to $X=0$ with probability $\frac{1}{k+1}$. For every $i\in [k]$, since $c_i(L\setminus F)\leq \frac{\varepsilon}{1+\varepsilon}\cdot c_i(L)$, we have that $c_i(L\setminus F)\leq \varepsilon\cdot c_i(L\cap F)$ and $c_i(L)\leq (1+\varepsilon)\cdot c_i(L\cap F)$. We observe that $F\setminus L$ contains a cut of the hypergraph $H- L$ and
    \[
        c_i(F\setminus L) = c_i(F) - c_i(F\cap L) \leq \frac{1}{1+\varepsilon}\cdot (b'_i - c_i(L))
        \qquad\text{for every }i\in[k],
    \]
    which implies that the instance $(H - L, \bar{c}|_{E-L}, \bar{b}'-\bar{c}(L))$ is $(1+\varepsilon)$-strongly budget-feasible. Moreover, PTAS-Algo on input $(H-L, \bar{c}|_{E-L}, \bar{b}'-\bar{c}(L))$ would branch at most $\ell-1$ times. 
    By the induction hypothesis, PTAS-Algo returns a budget-feasible cut of $(H - L, \bar{c}|_{E-L}, \bar{b}'-\bar{c}(L))$ with probability at least $q_{n,m-|L|}^{\ell-1}$. Moreover, if such a cut $F_2$ is returned, by \Cref{prop:ptas-multi-reductions}, $F_2\cup L$ contains a budget-feasible cut of $(H,\bar{c},\bar{b}')$ which would be returned by the algorithm. Therefore,
    \begin{align*}
        q^{\ell}_{n,m}
        &\ge
        \frac{1}{k+1}\cdot q^{\ell-1}_{n,m-|L|} \\
        &\ge \frac{1}{k+1}\cdot n^{-4k}\cdot \left(\frac{\varepsilon}{(k+1)\cdot (1+\varepsilon)}\right)^{\ell-1} \ \ \text{(by induction hypothesis)} \\
        &\ge n^{-4k}\cdot \left(\frac{\varepsilon}{(k+1)\cdot (1+\varepsilon)}\right)^{\ell}.
    \end{align*}

    \item There exists a hyperedge $e\in E$ with $|e|\geq |V|/(2k)$ and $c_i(L\setminus F)> \frac{\varepsilon}{1+\varepsilon}\cdot c_i(L)$ for some $i\in [k]$: In this case, PTAS-Algo branches to $X=i$ with probability $\frac{1}{k+1}$. PTAS-Algo contracts a hyperedge $e\in L$, which is sampled with probability proportional to its $c_i$-cost. In particular, the probability that $e\not\in F$ equals $\frac{c_i(L\setminus F)}{c_i(L)}$. Let $e\in L$ be the sampled hyperedge. PTAS-Algo would recurse on the instance $(H/e,\bar{c}|_{E-e},\bar{b}')$. 
    Suppose $e\not\in F$. 
    Then the instance $(H/e,\bar{c}|_{E-e},\bar{b}'))$ is $(1+\varepsilon)$-strongly budget-feasible: indeed, since $F=\delta(U)$ and $e\not\in \delta(U)$, we have that either $e\cap U=\emptyset$ or $e\subseteq U$; without loss of generality, suppose $e\cap U=\emptyset$, then $F=\delta(U)$ is a $(1+\varepsilon)$-strongly budget-feasible cut in  $(H/e,\bar{c}|_{E-e},\bar{b}'))$. 
    Moreover, PTAS-Algo on input $(H/e,\bar{c}|_{E-e},\bar{b}'))$ would branch at most $\ell-1$ times. 
    By the induction hypothesis, PTAS-Algo returns a budget-feasible cut of $(H/e,\bar{c}|_{E-e},\bar{b}')$ with probability at least $q_{n-|e|+1,m-1}^{\ell-1}$. Moreover, if such a budget-feasible cut $F_2$ is returned, then $F_2$ is also a budget-feasible cut of $(H,\bar{c},\bar{b}')$. Therefore,
    \begin{align*}
        q^{\ell}_{n,m}
        &\ge
        \frac{1}{k+1}\cdot \sum_{e\in L\setminus F}\frac{c_i(e)}{c_i(L)}\cdot q^{\ell-1}_{n-|e|+1,m-1} \\
        &\ge
        \frac{1}{k+1}\cdot \sum_{e\in L\setminus F}\frac{c_i(e)}{c_i(L)}\cdot (n-|e|+1)^{-4k}\cdot \left(\frac{\varepsilon}{(k+1)\cdot (1+\varepsilon)}\right)^{\ell-1} \\
        &\quad \quad \quad \quad \quad \quad \quad \quad \ \ \ \ \ \ \ \ \ \ \ \ \ \ \ \ \ \ \ \ \ \ \ \ \ \ \ \ \ \ \ \ \ \ \ \ \ \ \ \text{(by induction hypothesis)} \\
        &\ge
        \frac{1}{k+1}\cdot \sum_{e\in L\setminus F}\frac{c_i(e)}{c_i(L)}\cdot n^{-4k}\cdot \left(\frac{\varepsilon}{(k+1)\cdot (1+\varepsilon)}\right)^{\ell-1} \\
        &=
        \frac{1}{k+1}\cdot \frac{c_i(L\setminus F)}{c_i(L)}\cdot n^{-4k}\cdot \left(\frac{\varepsilon}{(k+1)\cdot (1+\varepsilon)}\right)^{\ell-1} \\
        &\ge \frac{1}{k+1}\cdot \frac{\varepsilon}{1+\varepsilon}\cdot n^{-4k}\cdot \left(\frac{\varepsilon}{(k+1)\cdot (1+\varepsilon)}\right)^{\ell-1} \ \ \text{(since $\frac{c_i(L\setminus F)}{c_i(L)}>\frac{\varepsilon}{1+\varepsilon}$)} \\
        &\ge n^{-4k}\cdot \left(\frac{\varepsilon}{(k+1)\cdot (1+\varepsilon)}\right)^{\ell}.
    \end{align*}
\end{enumerate}

Thus, inequality (\ref{inequality:multi-ptas-success-probability}) holds for every $(1+\varepsilon)$-strongly budget-feasible instance. According to \Cref{lemma:multi-fptas-depth}, $\ell\leq 8k\cdot \log n$, which further implies that one run of PTAS-Algo on instance $(H,\bar{c},\bar{b}')$ returns a budget-feasible cut with probability at least
\[
    n^{-4k}\cdot \left(\frac{\varepsilon}{(k+1)\cdot (1+\varepsilon)}\right)^{\ell} \geq n^{-4k}\cdot \left(\frac{\varepsilon}{(k+1)\cdot (1+\varepsilon)}\right)^{8k\cdot \log n}.
\]

\end{proof}

We now bound the runtime of PTAS-Algo.

\begin{lemma}\label{lemma:multi-ptas-runtime}
    $\text{PTAS-Algo} (H,\bar{c},\bar{b}')$ can be implemented to run in $(m+2^{8k})\cdot p^{O(1)}$ time, where $p$ and $m$ are the size and number of hyperedges of the input hypergraph.
\end{lemma}
\begin{proof}
    We emphasize that the recursion tree corresponding to PTAS-Algo is simply a path. Moreover, each recursive call reduces the number of hyperedges. Hence, the total number of nodes in the recursion tree is at most $m+1$. We observe that the last node with $n\leq 8k$ takes $2^{8k}\cdot p^{O(1)}$ time on the brute-force step, which implies that $\text{PTAS-Algo} (H,\bar{c},\bar{b}')$ takes $(m+2^{8k})\cdot p^{O(1)}$ time.
\end{proof}

All steps of PTAS-Algo except the brute-force search can be implemented to run in polynomial time---the only non-trivial step is Step 6(d)(iii), which can be implemented to run in polynomial time by Proposition \ref{prop:ptas-multi-reductions}(iv). Thus, \Cref{lem:multi-ptas-success,lemma:multi-ptas-runtime} together prove \Cref{thm:k-ptas}.

\section{Future Directions}\label{sec:conclusion}
We state some of the natural open questions raised by our work. Firstly, is it possible to solve multi-objective hypergraph min-cut in polynomial time? Or is it hard to design sub-quasi-polynomial-time algorithms under standard complexity-theoretic assumptions? 
Secondly, is it possible to bound the number of pareto-optimal points by a polynomial in the number of vertices? 
The instance that we discussed in the introduction to show that the number of optimal solutions for bi-objective hypergraph min-cut can be exponential has only one pareto-optimal point. Thirdly, is it possible obtain a PTAS for hypergraph connectivity interdiction? We have given a QPTAS via connections to bi-objective hypergraph budgeted-cut and Zenklusen's reduction \cite{Zen14}. It may be possible to obtain a PTAS without relying on tractability of bi-objective hypergraph budgeted-cut.

\paragraph{Acknowledgements.} Karthekeyan thanks Chao Xu and Calvin Beideman for preliminary discussions on the problem and for constructing counterexamples to certain approaches. For AI disclosure, see  \Cref{sec:ai-disc}.

\bibliographystyle{abbrv}
\bibliography{references}

@article{NI92,
author = {Nagamochi, Hiroshi and Ibaraki, Toshihide},
title = {Computing Edge-Connectivity in Multigraphs and Capacitated Graphs},
journal = {SIAM Journal on Discrete Mathematics},
volume = {5},
number = {1},
pages = {54-66},
year = {1992},
}

@inproceedings{FGKLS25,
  title={Fixed-parameter tractability of hedge cut},
  author={Fomin, F. V. and Golovach, P. A. and Korhonen, T. and Lokshtanov, D. and Saurabh, S.},
  booktitle={Proceedings of the 2025 Annual ACM-SIAM Symposium on Discrete Algorithms},
  pages={1402--1411},
  year={2025},
  series = {SODA},
}

@article{AHM2022,
  author  = {Adjiashvili, D. and Hommelsheim, F. and M\"{u}hlenthaler, M.},
  title   = {Flexible Graph Connectivity},
  journal = {Mathematical Programming},
  volume  = {192},
  pages   = {409--441},
  year    = {2022},
}

@INPROCEEDINGS{PY00,
  author={Papadimitriou, C.H. and Yannakakis, M.},
  booktitle={Proceedings 41st Annual Symposium on Foundations of Computer Science}, 
  series = {FOCS},
  title={On the approximability of trade-offs and optimal access of Web sources}, 
  year={2000},
  pages={86-92},
}

@article{Zen14,
author = {Zenklusen, R.}, 
title = {{Connectivity Interdiction}},
journal = {Operations Research Letters},
volume = {42},
year = {2014},
pages = {450--454},
}

@article{HOY26,
author = {Huang, CC. and Obscura Acosta, N. and Yingchareonthawornchai, S.}, 
title = {{An FPTAS for Connectivity Interdiction}},
journal = {Mathematical Programming},
volume = {216},
year = {2026},
pages = {605--626},
}

@inproceedings{KK15,
  title={Sketching cuts in graphs and hypergraphs},
  author={Kogan, Dmitry and Krauthgamer, Robert},
  booktitle={Proceedings of the 2015 Conference on Innovations in Theoretical Computer Science},
  pages={367--376},
  year={2015}
}

@article{NNI97,
 author = {Nagamochi, H. and Nishimura, K. and Ibaraki, T.}, 
 title = {Computing all small cuts in an undirected network}, 
 journal = {SIAM Journal on Discrete Mathematics},
 volume = {10},
 year = {1997}, 
 number = {3}, 
 pages = {469–481},
}

@inproceedings{BCW23,
author = {Beideman, C. and Chandrasekaran, K. and  Wang, W.},
title = {Approximate minimum cuts and their enumeration},
booktitle = {Symposium on Simplicity in Algorithms (SOSA)},
pages = {36-41},
year = {2023},
}

@inproceedings{Kar16,
author = {Karger, David R.},
title = {Enumerating parametric global minimum cuts by random interleaving},
year = {2016},
booktitle = {Proceedings of the Forty-Eighth Annual ACM Symposium on Theory of Computing},
pages = {542–555},
}

@inproceedings{AMR17,
  author       = {Aissi, H. and
                  Mahjoub, A. and
                  R. Ravi},
  title        = {Randomized Contractions for Multiobjective Minimum Cuts},
  booktitle    = {25th Annual European Symposium on Algorithms, {ESA}},
  pages        = {6:1--6:13},
  year         = {2017},
}

@article{BCX23,
author = {Beideman, C. and Chandrasekaran, K. and Xu, C.},
title = {Multicriteria cuts and size-constrained k-cuts in hypergraphs},
journal ={Mathematical Programming},
volume = {197}, 
pages = {27–69},
year = {2023},
}

@article{JLMPTS26,
author = {Jaffke, Lars and T. de Lima, Paloma and Masa\v{r}\'{\i}k, Tom\'{a}\v{s} and Pilipczuk, Marcin and Souza, Uev\'{e}rton S.},
title = {A Tight Quasi-Polynomial Bound for Global Label Min-Cut},
year = {2026},
volume = {22},
number = {2},
journal = {ACM Trans. Algorithms},
articleno = {23},
}

@inproceedings{Kar93,
  title={{Global Min-cuts in RNC, and Other Ramifications of a Simple Min-Cut Algorithm}},
  author={Karger, David R},
  booktitle={Proceedings of the fourth annual ACM-SIAM symposium on Discrete algorithms},
  series = {SODA},
  pages={21--30},
  year={1993}
}

@inproceedings{GKP17,
  title={Random contractions and sampling for Hypergraph and Hedge Connectivity},
  author={Ghaffari, M. and Karger, D. and Panigrahi, D.},
  booktitle={Proceedings of the Twenty-Eighth Annual ACM-SIAM Symposium on Discrete Algorithms},
  pages={1101--1114},
  series = {SODA},
  year={2017}
}

@article{CXY21,
	author = {Chandrasekaran, Karthekeyan and Xu, Chao and Yu, Xilin},
	title = {{Hypergraph $k$-Cut in randomized polynomial time}},
	journal = {Mathematical Programming},
	volume = {186}, 
	pages = {85-113}, 
	year = {2021},
}

@article{FPZ23,
  title={Minimum cut and minimum k-cut in hypergraphs via branching contractions},
  author={Fox, Kyle and Panigrahi, Debmalya and Zhang, Fred},
  journal={ACM Transactions on Algorithms},
  volume={19},
  number={2},
  pages={1--22},
  year={2023},
}

@article{AZ06,
title = {Multicriteria global minimum cuts},
author = {Armon, A. and Zwick, U.},
journal = {Algorithmica},
volume = {46},
number = {1},
pages = {15–26},
year = {2006},
}

@article{BCW23-enumhkcut-j,
	author = {Beideman, Calvin and Chandrasekaran, Karthekeyan and Wang, Weihang},
	title = {{Deterministic enumeration of all minimum $k$-cut-sets in hypergraphs for fixed $k$}},
	year = {2024},
	journal = {Mathematical Programming},
	volume = {207},
	pages = {329--367},
	doi = {10.1007/s10107-023-02013-8},
}

@article{AMMQ15,
	title={Strongly polynomial bounds for multiobjective and parametric global minimum cuts in graphs and hypergraphs},
	author={Aissi, H. and Mahjoub, A. and McCormick, T. and Queyranne, M.},
	journal={Mathematical Programming},
	volume={154},
	number={1-2},
	pages={3--28},
	year={2015},
}

@article{CX18,
  title={Minimum cuts and sparsification in hypergraphs},
  author={Chekuri, C. and Xu, C.},
  journal={SIAM Journal on Computing},
  volume={47},
  number={6},
  pages={2118--2156},
  year={2018},
  publisher={SIAM}
}

@article{Que98,
  title={Minimizing symmetric submodular functions},
  author={Queyranne, M.},
  journal={Mathematical Programming},
  volume={82},
  pages={3--12},
  year={1998},
}

@article{CQX20,
author = {Chekuri, C. and Quanrud, K. and Xu, C.},
title = {{LP} Relaxation and Tree Packing for Minimum $k$-Cut},
journal = {SIAM Journal on Discrete Mathematics},
volume = {34},
number = {2},
pages = {1334-1353},
year = {2020},
}

@article{li2023polylogarithmic,
  title={{Polylogarithmic Approximation for Robust $s$-$t$ Path}},
  author={Li, Shi and Xu, Chenyang and Zhang, Ruilong},
  journal={arXiv:2305.16439},
  year={2023}
}

@inproceedings{chekuri2026polylogarithmic,
  title={A Polylogarithmic Approximation for Buy-at-Bulk Network Design with Protection},
  author={Chekuri, Chandra and Jain, Rhea},
  booktitle={Proceedings of the 58th Annual ACM Symposium on Theory of Computing},
  pages={2230--2241},
  year={2026}
}

@inproceedings{hershkowitz2026planar,
  title={Planar Length-Constrained Minimum Spanning Trees},
  author={Hershkowitz, D Ellis and Huang, Richard Z},
  booktitle={Proceedings of the 58th Annual ACM Symposium on Theory of Computing},
  pages={1528--1539},
  year={2026}
}

@inproceedings{hop_distance21,
author = {Haeupler, Bernhard and Hershkowitz, D. Ellis and Zuzic, Goran},
title = {Tree Embeddings for Hop-Constrained Network Design},
year = {2021},
isbn = {9781450380539},
doi = {10.1145/3406325.3451053},
booktitle = {Proceedings of the 53rd Annual ACM SIGACT Symposium on Theory of Computing},
pages = {356-369},
series = {STOC}
}

@INPROCEEDINGS {filtser22,
author = {A. Filtser},
booktitle = {2021 IEEE 62nd Annual Symposium on Foundations of Computer Science (FOCS)},
title = {Hop-Constrained Metric Embeddings and their Applications},
year = {2022},
volume = {},
issn = {},
pages = {492-503},
doi = {10.1109/FOCS52979.2021.00056},
url = {https://doi.ieeecomputersociety.org/10.1109/FOCS52979.2021.00056},
publisher = {IEEE Computer Society},
address = {Los Alamitos, CA, USA},
month = feb
}

\appendix

\section{Single-objective Algorithm}\label{sec:uni}

We focus on the single-objective budgeted-cut problem as a warm-up towards the $k$-objective budgeted-cut problem.
We state the central problem of interest to this section now. Let $H=(V, E)$ be a hypergraph with non-negative hyperedge-costs $c: E\rightarrow \R_{\ge 0}$ and budget $b\in \R_{\ge 0}$. A subset $F\subseteq E$ is a budget-feasible cut for $(H, c, b)$ if (1) $F$ is a cut, i.e., $F=\delta_H(U)$ for some non-empty proper subset $U$ of vertices and (2) $c(F)\le b$. 

\begin{center}
\fbox{%
\begin{minipage}{0.92\textwidth}
\textbf{Single-objective Hypergraph Budgeted-Cut.}

\textbf{Given:} Hypergraph \(H=(V,E)\), costs \(c:E\to\R_{\ge0}\), and a budget \(b\in \R_{\ge 0}\).

\textbf{Goal:} Return a budget feasible cut for $(H, c,b)$ if one exists. 
\end{minipage}}
\end{center}
The following is the main result of this section. 
\begin{restatable}{theorem}{UniTheorem}\label{thm:uni}
There is a randomized algorithm that takes as input 
an instance $(H, c,b)$ of Single-Objective Hypergraph Budgeted-Cut and has the following properties:
(i) if the instance has a budget-feasible cut, then the algorithm returns a budget-feasible cut with probability at least $\frac{1}{\binom{n}{2}}$ and otherwise, returns FAIL and
(ii) the expected run-time is $n^2p^{O(1)}$, where $p$ and $n$ are the size and number of vertices of the input hypergraph. 
\end{restatable}
Repeating the algorithm \(O(n^2\log(1/\eta))\) times finds a budget-feasible cut with probability at least \(1-\eta\), whenever one exists, in polynomial expected time.

\begin{algorithm}[H]
\caption{Single-objective-Algo\((H=(V,E),c:E\rightarrow \R_{\ge 0},b\in \R_{\ge 0})\)}
\begin{enumerate}[leftmargin=2em]
\item Remove all loops from \(H\).
\item Let $D:=\{e\in E:e=V\}$. Set
\begin{align*}
        H&\leftarrow H-D \text{ and}\\
        b&\leftarrow b-c(D).
\end{align*}
If \(b<0\), return \(\FAIL\).
\item If there exists \(v\in V\) such that \(c(\cut_H(\{v\}))\le b\), return \(\cut_H(\{v\})\cup D\).
\item Sample a hyperedge \(e\) with probability
\[
        \Pr[e]=\frac{c(e)}{c(E)}.
\]
\item \(F_1\leftarrow \text{Single-objective-Algo}(H/e,c|_{E-e},b)\).
\item If \(F_1\) is not \(\FAIL\), return \(F_1\cup D\).
\item With probability \(|e|/|V|\):
\begin{enumerate}[leftmargin=2em]
    \item \(F_2'\leftarrow \text{Single-objective-Algo}(H\setminus e,c|_{E-e},b-c(e))\).
    \item If \(F_2'\) is not \(\FAIL\), find a non-empty proper subset $U\subseteq V$ with $\delta_H(U)\subseteq F_2'\cup \{e\}$ and return $\delta_H(U)\cup D$.
\end{enumerate}
\item Return \(\FAIL\).
\end{enumerate}
\end{algorithm}

See Algorithm 3 for a pseudocode of our algorithm to prove Theorem \ref{thm:uni}. 
The sampling distribution in Step~4 is well-defined whenever that step is reached:  After Step~3, at least one remaining hyperedge has positive cost. Indeed, if no hyperedges have positive cost, then every singleton cut would have cost zero and Step~3 would have returned a budget-feasible cut, since Step~2 already checked that the residual budget is nonnegative.

The algorithm is intentionally sequential.  It always executes the contraction branch first.  The deletion branch is sampled only if the contraction branch fails. This convention helps in the success probability while limiting the number of recursive calls: if the contraction call succeeds, the algorithm is already done and if the contraction call fails, then the deletion branch is executed in a probabilistic manner.

The algorithm deals with the recursive output from the contraction call and the deletion call slightly differently (Steps 6 and 7(b)). The output of the contraction call in Step 6 is returned directly while the output of the deletion call in Step 7(b) is refined before returning (both after  taking the union with $D$). The reason for refining in Step 7(b) is to ensure that the returned set is indeed a cut and not just the superset of a cut.

\subsection{Correctness}

We now analyze Algorithm~3. 
We first record the deterministic and recursive preservation facts that will help in proving correctness. See proof of \Cref{prop:multi-reductions} for a proof of the following facts. 

\begin{proposition}\label{prop:uni-reductions}
Let \((H,c,b)\) be the current instance.
\begin{enumerate}[label=(\roman*),leftmargin=2em]
\item Deleting loops does not affect budget-feasibility.
\item Let \(D\) be the set of spanning hyperedges. 
If \((H,c,b)\) has a budget-feasible cut, then $b-c(D)\ge 0$ and the instance \((H-D,c|_{E-D},b-c(D))\) has a budget-feasible cut.  Conversely, if $b-c(D)\ge 0$ and \(F'\) is a budget-feasible cut in \((H-D,c|_{E-D},b-c(D))\), then \(F'\cup D\) is a budget-feasible cut in \((H,c,b)\).
\item For $e\in E$, if \(F_1\) is a budget-feasible cut in $(H/e, c|_{E-e}, b)$, then \(F_1\) is a budget-feasible cut in $(H, c, b)$.
\item For $e\in E$, given a budget-feasible cut \(F_2'\) in $(H\setminus e,c|_{E-e},b-c(e))$, there exists a non-empty proper subset $U\subseteq V$ with  \(\delta_H(U)\subseteq F_2'\cup\{e\}\) and $\delta_H(U)$ is a budget-feasible cut in $(H,c,b)$ and such a subset $U$ can be found in time $n^{O(1)}$.
\end{enumerate}
\end{proposition}

The lemma below follows from Lemma \ref{lem:uni-singleton-infeasibility-inequality} (apply it to $(H, c, b)$). 

\begin{lemma}[Singleton Infeasible Set Inequality for Single-objective]\label{lem:uni-singleton-infeasibility-ineq}
Consider a loopless input instance $(H=(V, E), c: E\rightarrow \R_{\ge 0}, b\in \R_{\ge 0})$ that has a budget-feasible cut $F\subseteq E$. Suppose $c(\delta_H(v))>b$ for every $v\in V$. 
Then, 
\[
        \sum_{e\in E\setminus F} c(e)|e|
        >
        \sum_{e\in F} c(e)|V\setminus e|. 
\]
\end{lemma}
We need the following binomial inequality. 
\begin{proposition}\label{prop:uni-numeric}
For integers \(n\ge3\) and \(2\le a\le n-1\),
\[
        \frac{\binom n2}{\binom{n-a+1}{2}}-1
        \ge
        \frac{a}{n}.
\]
\end{proposition}

\begin{proof}
The left-hand side is
\[
        \frac{n(n-1)}{(n-a+1)(n-a)}-1
        =
        \frac{(a-1)(2n-a)}{(n-a+1)(n-a)}.
\]
Thus it suffices to show
\[
        n(a-1)(2n-a)
        \ge
        a(n-a+1)(n-a).
\]
The difference between the two sides equals
\[
        (a-2)n^2+a(a-1)n+a^2(n-a),
\]
which is nonnegative for \(2\le a\le n-1\).
\end{proof}

We now analyze the success probability. 

\begin{lemma}[Single-objective success probability]\label{lem:uni-success}
Every set returned by the algorithm is budget-feasible.
Moreover, if the input instance \((H,c,b)\) admits a budget-feasible cut and \(|V|=n\), then one run of Algorithm 3 returns a budget-feasible cut with probability at least
\[
        \binom n2^{-1}.
\]
\end{lemma}

\begin{proof}
Every returned set is obtained from a budget-feasible cut returned by a recursive call, or is a singleton cut that is verified explicitly in Step~3. Hence, by \Cref{prop:uni-reductions}, every set returned by the algorithm is a budget-feasible cut. We next analyze the success probability. 

Let $q^m_n$ be the least probability of returning a budget-feasible cut of the algorithm over all input instances with $n$ vertices and $m$ hyperedges that have a budget-feasible cut. We will show that $q^m_n\ge \binom{n}{2}^{-1}$ by induction on $n+m$. The base case is $n+m=2$ and in particular $n=2$ and $m=0$. 
In this case, since the instance has a budget-feasible cut, it should be the case that the budget is non-negative and hence, the singleton cut is budget-feasible; the algorithm indeed returns a budget-feasible singleton cut with probability $1$. 
The rest of the proof shows the induction step. 

Fix a $n$-vertex $m$-hyperedge input instance $(H, c, b)$ that has a budget-feasible cut. If $H$ has a loop or a spanning hyperedge,  
then the claim follows using the induction hypothesis for $n+m-1$ and by \Cref{prop:uni-reductions}. Hence, we may assume that the input instance is loopless and has no spanning hyperedges. 
If it has a vertex $v\in V$ such that $c(\delta(v))\le b$, then the algorithm returns a budget-feasible cut with probability $1$. Hence, we may assume that the input instance is loopless, has no spanning 
hyperedges, and has no budget-feasible singleton cut. Thus, the algorithm reaches Step~4. 

At least one hyperedge in $H$ has positive cost (since the algorithm reaches Step~4, and hence, the residual budget is non-negative and none of the singleton cut is budget-feasible). Consequently, the sampling distribution in Step~4 is well-defined. Moreover, the instance is loopless and 
has no spanning hyperedges.  Therefore every sampled hyperedge has size \(a\in\{2,\ldots,n-1\}\).  In particular, the random step can occur only when $n\ge 3$. 

Let $F$ be a budget-feasible cut in $(H, c, b)$. Let $e\in E $ be the sampled hyperedge. We consider two cases.
\begin{enumerate}
\item Suppose $e\not\in F$. Then, $(H/e, c|_{E-e}, b)$ has a budget-feasible cut by Proposition \ref{prop:uni-reductions} and moreover, has $n-|e|+1$ vertices and at most $m$ hyperedges. Thus, the algorithm executed on the contraction instance $(H/e, c|_{E-e}, b)$ succeeds in returning a budget-feasible cut in $(H/e,c|_{E-e},b)$ with probability at least $q^m_{n-|e|+1}$ and by Proposition \ref{prop:uni-reductions}, the algorithm returns a budget-feasible cut in $(H, c, b)$ with probability at least $q^m_{n-|e|+1}$. 
\item Suppose $e\in F$. Then, $(H\setminus e, c|_{E-e}, b-c(e))$ has a budget-feasible cut by Proposition \ref{prop:uni-reductions} and moreover, has $n$ vertices and $m-1$ hyperedges. 
The algorithm first tries the contraction call.  If it succeeds, then the algorithm is already successful.  If it fails, then the deletion branch is executed with probability $|e|/n$, and, conditioned on executing the deletion branch, succeeds with probability at least $q^{m-1}_n$.  Consequently, the conditional success probability in this case is at least $(|e|/n)q^{m-1}_n$.
\end{enumerate}

Consequently, the success probability \(q^m_n\) satisfies
\[
        q^m_n
        \ge
        \sum_{e\in E\setminus F}\left(\frac{c(e)}{c(E)}\right)q^m_{n-|e|+1}
        +
        \sum_{e\in F}\left(\frac{c(e)}{c(E)}\right)\left(\frac{|e|}{n}\right)q^{m-1}_n.
\]
We show that the right-hand side is at least \(\binom{n}{2}^{-1}\). By induction hypothesis, we know that $q^{m}_{n-|e|+1}\ge \binom{n-|e|+1}{2}^{-1}$ and $q^{m-1}_n\ge \binom{n}{2}^{-1}$. For ease of notation, let $p_e:=c(e)/c(E)$ for each $e\in E$. 
Substituting, we obtain that 
\begin{align*}
    q^m_n 
    &\ge \sum_{e\in E\setminus F}\frac{p_e}{\binom{n-|e|+1}{2}} + \sum_{e\in F}\left(\frac{|e|}{n}\right)\frac{p_e}{\binom{n}{2}}\\
    &\ge \sum_{e\in E\setminus F}\frac{p_e}{\binom{n}{2}}\left({1+\frac{|e|}{n}}\right) + \sum_{e\in F}\left(\frac{|e|}{n}\right)\frac{p_e}{\binom{n}{2}} \quad \quad \text{(by \Cref{prop:uni-numeric})}\\
    &=\frac{1}{\binom{n}{2}}\left(\sum_{e\in E\setminus F}p_e\left({1+\frac{|e|}{n}}\right)+\sum_{e\in F}p_e\frac{|e|}{n}\right)\\
    &=\frac{1}{\binom{n}{2}}\left(\sum_{e\in E}p_e + \sum_{e\in E\setminus F}p_e\left(\frac{|e|}{n}\right)-\sum_{e\in F}p_e\left(1-\frac{|e|}{n}\right)\right)\\
    &\ge \frac{1}{\binom{n}{2}}\left(1 + \sum_{e\in E\setminus F}p_e\left(\frac{|e|}{n}\right)-\sum_{e\in F}p_e\left(1-\frac{|e|}{n}\right)\right).
\end{align*}
It suffices to show that 
\[
\sum_{e\in E\setminus F}p_e\left(\frac{|e|}{n}\right)\ge \sum_{e\in F}p_e\left(1-\frac{|e|}{n}\right). 
\]
Substituting $p_e =c(e)/c(E)$ and $n=|V|$, it suffices to show that 
\[
\sum_{e\in E\setminus F}c(e)|e|\ge \sum_{e\in F}c(e)|V\setminus e|,
\]
which holds by \Cref{lem:uni-singleton-infeasibility-ineq}. 
\end{proof}
\begin{remark}[Why this does not imply a counting bound]\label{rem:uni-not-counting}
\Cref{lem:uni-success} does not imply a bound on the number of budget-feasible cuts in a hypergraph (or even in a graph). 
The proof of \Cref{lem:uni-success} is not target-specific: 
the algorithm may not return a fixed budget-feasible cut $F$. 
In particular, the algorithm may switch to an easier budget-feasible cut in some recursive state owing to budget-feasibility of some singleton cut.
Therefore, the analysis does not imply a cut-counting bound. 
Besides, it is easy to construct a graph instance $(G, c, b)$ in which the number of budget-feasible cuts is exponential in the number of vertices: consider a complete graph with unit edge costs and budget $b=\binom{n}{2}$. 
\end{remark}

\subsection{Runtime}
We bound the expected number of recursive calls made by the algorithm. 
\begin{lemma}[Single-objective recursion tree]\label{lem:uni-runtime}
The expected number of nodes in the recursion tree of 
Single-objective-Algo on an input instance with $n$ vertices is at most 
\[
\frac{n(3n-5)}{2}. 
\]
\end{lemma}
\begin{proof}
Let \(R(n,m)\) be the worst-case expected number of nodes in the recursion tree of Single-objective-Algo for an input instance that has \(n\) vertices and \(m\) hyperedges. We will show the following by induction on $n+m$:
\[
R(n, m) \le \frac{n(3n-5)}{2}. 
\]

The base case corresponds to $n+m=2$ and in particular $n=2$ and $m=0$. In this case, the algorithm terminates without any recursive calls and hence $R(n,m)=1$. We now show the induction step. 

Let $(H, c, b)$ be an input instance with $n$ vertices and $m$ hyperedges. If there exists a singleton cut that is budget-feasible, then $R(n, m)\le 1$. Suppose no singleton cut is budget-feasible. Then, $n\ge 3$. Let $e$ be a sampled hyperedge. We note that $|e|\ge 2$. 
The contraction branch instance has \(n-|e|+1\) vertices because \(|e|\) vertices are identified into one.  The deletion branch instance has the same number of vertices, but one fewer hyperedge. Hence,
\begin{align*}
    &E\left[\text{Number of recursive calls for }(H, c, b)|e\text{ is sampled}\right]\\
    &\quad \quad \le 1+R(n-|e|+1,m) + \frac{|e|}{n}R(n, m-1)\\
    &\quad \quad \le 1+\frac{(n-|e|+1)(3(n-|e|+1)-5)}{2} + \frac{|e|}{n}\left(\frac{n(3n-5)}{2}\right) \quad \quad \text{(by induction hypothesis)}\\
    &\quad \quad = \frac{(n-|e|)(3n-3|e|+1)}{2} + \frac{|e|(3n-5)}{2}\\
    &\quad \quad \le \frac{(n-|e|)(3n-5)}{2} + \frac{|e|(3n-5)}{2} \quad \quad \text{(since $|e|\ge 2$)}\\
    &\quad \quad = \frac{n(3n-5)}{2}. 
\end{align*}
Thus, $R(n,m)\le n(3n-5)/2$. 
\end{proof}

All steps of Single-Objective-Algo can be implemented to run in polynomial time---the only non-trivial step is Step 7(b), which can be implemented to run in polynomial time by Proposition \ref{prop:uni-reductions}(iv). Thus, 
\Cref{lem:uni-success,lem:uni-runtime} together prove Theorem \ref{thm:uni}.  

\begin{remark}\label{rem:hedgegraph-results}
The algorithm in this section leads to a new algorithm that yields two new results for the \emph{hedgegraph min-cut problem}.
We recall that a \emph{hedgegraph} is specified by a vertex set and a collection of hedges, where each hedge is a vertex-disjoint collection of hyperedges. In the hedgegraph min-cut problem, we are given a hedgegraph with a non-negative cost function $c$ on the hedges, and the goal is to find a minimum-cost subset of hedges whose removal disconnects the hedgegraph. Now, consider the budgeted variant of the problem where we are also given a budget $b\in \Z_{\ge 0}$ and the goal is to verify whether there exists a hedge-cut of cost at most $b$. For a hedge $e=\{h_1, h_2, \ldots, h_s\}$, where $h_1, h_2, \ldots, h_s$ are the vertex disjoint hyperedges constituting $e$, we define size$(h):=|\cup_{i=1}^s h_i|$. The algorithm is identical to Algorithm 3 except for Step 7: after sampling a hedge $e$, instead of exploring the deletion branch with probability $|e|/|V|$, we would explore the deletion branch with probability $\text{size}(e)/|V|$. We can show that the success probability of this algorithm is $\binom{n}{2}^{-1}$ while the expected number of nodes in the recursion tree is $\min\{m^{O(\log{n})}, 2^{O(b)}n\}$. Thus, we obtain (i) a quasi-polynomial runtime for hedgegraph min-cut as well as (ii) a fixed-parameter algorithm parameterized by the solution size for hedgegraph min-cut in the unit-cost setting. 
We note that our run-time is strongly quasi-polynomial time while the known quasi-polynomial time algorithm of Ghaffari, Karger, and Panigrahi \cite{GKP17} is only weakly quasi-polynomial time (i.e., depends on the cost function). Moreover, our algorithm is relatively simpler and more natural than the known FPT parameterized by the solution size \cite{FGKLS25}. Another interesting feature is that the same algorithm leads to both results. 
\end{remark}

\end{document}